\documentclass[11pt]{amsart}
\usepackage{amscd,amssymb,amsmath,latexsym,enumerate}
\usepackage{bbm}
\usepackage[mathscr]{euscript}
\usepackage{tikz}
\usepackage{tikz-cd}
\usepackage{epsfig} 
\usepackage{stackengine}
\numberwithin{equation}{section}
\usepackage{hyperref}

\usepackage{color}

\newtheorem{theo}{Theorem}[section]

\newtheorem{defini}[theo]{Definition}
\newtheorem{proposi}[theo]{Proposition}
\newtheorem{lemma}[theo]{Lemma}
\newtheorem{coro}[theo]{Corollary}
\newtheorem{exam}[theo]{Example}

\theoremstyle{remark}
\newtheorem*{rem}{Remark}

\newcommand{\Aa}{{\mathcal A}}

\newcommand{\Ee}{{\mathcal E}}

\newcommand{\Gg}{{\mathcal G}}
\newcommand{\Hh}{{\mathcal H}}

\newcommand{\Jj}{{\mathcal J}}
\newcommand{\Kk}{{\mathcal K}}
\newcommand{\Ll}{{\mathcal L}}

\newcommand{\Tt}{{\mathcal T}}

\newcommand{\Vv}{{\mathcal V}}

\newcommand{\CM}{{\mathbb C}}

\newcommand{\N}{{\mathbb N}}
\newcommand{\QM}{{\mathbb Q}}

\newcommand{\RM}{{\mathbb R}}

\newcommand{\ZM}{{\mathbb Z}}

\newcommand{\bs}{{\mathscr B}}

\newcommand{\js}{{\mathscr J}}

\newcommand{\ws}{{\mathscr W}}

\newcommand{\Orb}{\textrm{Orb}}                 
\newcommand{\tra}{\mbox{\sc t}}                 
\newcommand{\Inv}{\js}                  		
\newcommand{\dic}{\ws}               			
\newcommand{\scf}{\wp}	  						
\newcommand{\csp}{\Aa^\ZM}	  					
\newcommand{\QS}{{\mathfrak S}}                 

\renewcommand{\c}{c}
\renewcommand{\d}{d}
\newcommand{\ca}{c_s}
\newcommand{\cb}{c_l}

\newcommand{\dist}{\textrm{dist}}

\renewcommand{\mu}{\mathrm{Leb}}

\definecolor{grau}{rgb}{0.4,0.4,0.4}

\definecolor{green}{RGB}{0, 180, 0}
\definecolor{cyan}{RGB}{0, 180, 180}
\definecolor{yellow}{RGB}{211,211,0}

\renewcommand{\subset}{\subseteq}

\newcommand{\ol}[1]{\overline{#1}}

\usepackage[
   backend=bibtex,
   maxcitenames=6,
   maxbibnames=6,
   maxalphanames = 5,
   style=alphabetic]{biblatex}
\begin{document}

\title[Spectral regularity and defects]{Spectral regularity and defects for the Kohmoto model}

\author{Siegfried Beckus, Jean Bellissard, Yannik Thomas}

\address{Institut für Mathematik\\
Universit\"at Potsdam\\
Potsdam, Germany}
\email{beckus@uni-potsdam.de}

\address{Georgia Institute of Technology\\
School of Mathematics\\
Atlanta GA, USA}
\email{jeanbel@me.com}

\address{Institut für Mathematik\\
Universit\"at Potsdam\\
Potsdam, Germany}
\email{yathomas@uni-potsdam.de}

\begin{abstract}
We study the Kohmoto model including Sturmian Hamiltonians and the associated Kohmoto butterfly. 
We prove spectral estimates for the operators using Farey numbers. 
In addition, we determine the impurities at rational rotations leading to the spectral defects in the Kohmoto butterfly.
Our results are similar to the ones obtained for the Almost-Mathieu operator and the associated Hofstadter butterfly.
\end{abstract}


\maketitle
\tableofcontents


\section{Introduction and main results}
\label{sect-MainResult}

For $\alpha\in[0,1]$ and a coupling constant $V\neq 0$, consider the self-adjoint bounded operators $H_\alpha:=H_{\alpha,V}:\ell^2(\ZM)\to\ell^2(\ZM)$ defined by
\begin{equation}
\label{eq:KohmOperat}
\big(H_{\alpha,V} \psi\big)(n) 
	= \psi(n-1) + \psi(n+1) + V\, \chi_{[1-\alpha,1)}(n\alpha\mod\, 1) \cdot \psi(n)
	\,,\qquad \psi\in\ell^2(\ZM)\,,\; n\in\ZM.
\end{equation}
This family of operators $H_\alpha$ for all $\alpha\in[0,1]$ was introduced in \cite{KKT83} and is called {\em Kohmoto model}. 
The operators $H_\alpha$ for $\alpha\in[0,1]\setminus\QM$ irrational are also called Sturmian Hamiltonians. 
If $\alpha\in[0,1]$ is rational, $H_\alpha$ is periodic and so its spectral properties can be computed using Floquet-Bloch theory. 
The spectra of these operators with rational rotation $\alpha$ are plotted in Figure~\ref{fig-Kohmoto} -- a graphic we call {\em Kohmoto butterfly}.

In order to get access to the irrational case, rational approximations have been extensively used in the literature where the ground was laid in the works \cite{Cas86,Sut89,BIST89,BIT91,Raym95}. 
The Sturmian Hamiltonians have Cantor spectrum of Lebesgue measure zero \cite{Sut89,BIST89} for all $V>0$. A hierarchical structure of the rational approximations was developed in \cite{Raym95} (see also the review \cite{BaBeBiRaTh24}) for large couplings $V>4$ and recently extended in \cite{BBL24} to $V\neq 0$. This allows to compute the integrated density of states \cite{Raym95,BBL24} and estimate the Hausdorff dimension, see e.g. \cite{Raym95,KiKiLa03,DaEmGoTc08,LiQuWe14,DaGoYe16} and references therein. 
The reader is referred to \cite{Dam07_survey,DaEmGo15-survey,Dam17-Survey,DaFi24-book_2} for a more detailed elaboration. 

\begin{figure}[htb]
 \centering
	\includegraphics[scale=0.71]{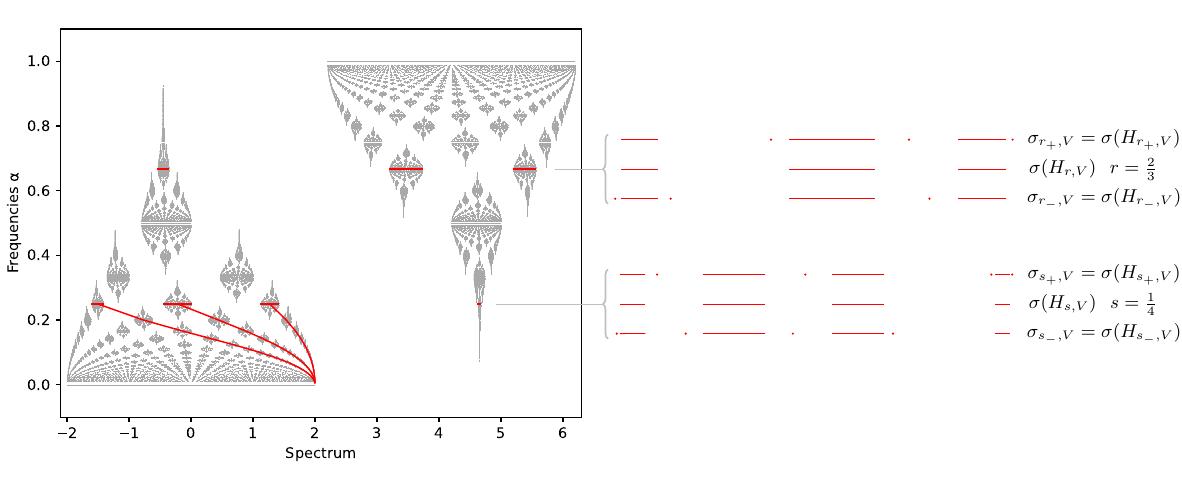}
 \caption{A plot of the Kohmoto butterfly for $V>4$. The spectrum $\sigma(H_{r,V})$ at each $r\in[0,1]\cap\QM$ is split up into three points: the lower limit $r_-$, the upper limit $r_+$ and $r$, confer Proposition~\ref{prop-Defects}. Proposition~\ref{prop:DiscreteSpectrum=q_points} asserts that the upper and lower limiting spectra have exactly one additional (compared to  $\sigma(H_{r,V})$) point in each bounded spectral gap plus one point in one of the unbounded spectral gaps of $\sigma(H_{r,V})$.
This is demonstrated on the right hand side for $r=\frac{2}{3}$ and $s=\frac{1}{4}$.
The red lines indicate the closing of the spectral gaps which is used to prove the optimality of the spectral estimates in Theorem~\ref{theo-SpeCon}.
 }
 \label{fig-Kohmoto}
\end{figure}

A first crucial step is to observe that the spectra $\sigma(H_{\alpha_n,V})$ of rational approximations $\alpha_n$ of $\alpha\in[0,1]\setminus\QM$ converge to the spectrum $\sigma(H_{\alpha,V})$ if $\alpha_n\to \alpha$. 
Here the convergence of the spectrum is measured in the Hausdorff metric $d_H$ on the set $\Kk(\RM)$ of compact subsets of $\RM$. 
In \cite{BIT91}, the authors used $C^\ast$-algebraic methods to prove that the spectral map for $V\neq 0$,
$$
\Sigma_V:[0,1]\to \Kk(\RM), \quad \alpha\mapsto \sigma(H_{\alpha,V}),
$$
is continuous at the irrational points $\alpha$ and discontinuous at the rational points $\alpha$, where $[0,1]$ is equipped with the Euclidean topology.
If $r=\frac{p}{q}\in[0,1]$ with $p,q$ coprime, then $\Sigma_V(r)$ consists of exactly $q$-intervals and the connected components of $\RM\setminus\Sigma_V(r)$ are called {\em spectral gaps}. 
Already from the Kohmoto butterfly in Figure~\ref{fig-Kohmoto}, the discontinuity can be numerically observed if $\alpha$ is rational, see also \cite{BIT91}. 
For instance, at $r=\frac{2}{3}$ or $s=\frac{1}{4}$, we observe $3$ or  $4$ points in the spectral gaps if we approach the rational number from above or below respectively. 
These points obtained in the spectral gaps are called {\em spectral defects}.

Zooming out again, we address the following questions in this work.
\begin{enumerate}[(Q1)]
\item \label{enu:Q1_Lipschitz} What is the right metric on $[0,1]$ making the map $\Sigma_V$ not only continuous but even Lipschitz continuous?  - Theorem~\ref{theo-SpeCon}
\item \label{enu:Q2_Optimal} Are the obtained spectral estimates optimal? - Theorem~\ref{theo-Optimality}
\item \label{enu:Q3_Defects} Can one describe these impurities at rational points explicitly on the level of the underlying dynamical system? - Proposition~\ref{prop-Defects}
\item \label{enu:Q4_OrderingPoints} How many spectral defects are generated in the spectral gaps of $\Sigma_V(r)$ for $r\in[0,1]\cap\QM$, if it is approximated from above or below and what is their relative ordering? - Proposition~\ref{prop:DiscreteSpectrum=q_points}
\end{enumerate}

Our starting point are recent developments \cite{BBdN18,BBC18,BBdN20,BT21} proving that the spectral convergence is tightly connected to the convergence of the underlying dynamical systems.

We prove that the Farey numbers (introduced in Section~\ref{ssect-FareyNumbers}) play a central role here to define a new metric $d_F$ on the interval $[0,1]$. 
Then we show that the spectral map $\Sigma_V$ extends uniquely to the completion $\ol{[0,1]}_F$ of $[0,1]$ with respect to $d_F$. 
In order to do so, we isometrically represent the Farey space $(\ol{[0,1]}_F, d_F)$ as the boundary of the interval Farey tree and as a suitable subspace of dynamical systems - the space of {\em mechanical systems} (see Section~\ref{ssect-subshifts}). These representations allow us to answer the previously raised questions.

Let us shortly explain the dynamical perspective. 
The product space $\{0,1\}^\ZM:=\big\{\omega:\ZM\to\{0,1\}\big\}$ of two-sided infinite sequences with zeros and ones is a compact metric space equipped with a $\ZM$-action by shifting the origin of a configuration $\omega\in\{0,1\}^\ZM$, see details in Section~\ref{ssect-subshifts}. 
The potentials in the Kohmoto model are defined using the configurations 
$$
\omega_\alpha\in\{0,1\}^\ZM, \qquad
	\omega_\alpha(n):=\chi_{[1-\alpha,1)}(n\alpha\mod\, 1),
	\qquad n\in\ZM,\, \alpha\in[0,1].
$$
For rational $r=\frac{p}{q}\in[0,1]$, the configuration $\omega_r$ is $q$-periodic, namely $\omega_r(n+q)=\omega_r(n)$ for all $n\in\ZM$.
We have discussed above that the spectral map $\Sigma_V$ is discontinuous at the rational points $r\in[0,1]\cap\QM$ and in fact the Kohmoto butterfly (Figure~\ref{fig-Kohmoto}) hints that we obtain two different limits in the Hausdorff metric of the spectrum
$$
\sigma_{r_+}(V):=\lim_{\alpha\searrow r} \Sigma_V(\alpha) 
	\qquad\textrm{and}\qquad
\sigma_{r_-}(V):=\lim_{\alpha\nearrow r} \Sigma_V(\alpha).
$$
Using $C^\ast$-algebraic methods, the authors in \cite{BIT91} proved that the difference between $\sigma_{r_\pm}(V)$ and $\Sigma_V(r)$ are finitely many points -- spectral defects generated in the spectral gaps of $\Sigma_V(r)$. 
We will construct for each rational $r\in[0,1]\cap\QM$, two sequences $\omega_{r_+},\omega_{r_-} \in\{0,1\}^\ZM$ such that $\sigma_{r_\pm}(V)=\sigma(H_{\omega_{r_\pm,V}})$ where the self-adjoint operator $H_{\omega,V}:\ell^2(\ZM)\to\ell^2(\ZM)$ for $\omega\in\{0,1\}^\ZM$ and $V\neq 0$ is defined by
\begin{equation}
\label{eq:SchroedOper}
\big(H_{\omega,V} \psi\big)(n) 
	= \psi(n-1) + \psi(n+1) + V\, \omega(n) \cdot \psi(n)
	\,,\qquad \psi\in\ell^2(\ZM)\,,\; n\in\ZM.
\end{equation}
Interestingly these sequences $\omega_{r_+}$ and $\omega_{r_-}$ are perturbations of the sequence $\omega_r$ adding a localized impurity around the origin, which is determined by the corresponding Farey neighbors (defined in Section~\ref{ssect-FareyNumbers}) of the rational number $r$, see Proposition~\ref{prop-Defects}. 
With this at hand, we recover the result of \cite{BIT91} stating that $\sigma_{\mathrm{ess}}(H_{\omega_{r_{\pm}},V})=\sigma_{\mathrm{ess}}(H_{\omega_r,V})$ for the essential spectrum, see Proposition~\ref{prop:EssentialSpectrum}.

On the one hand, this last result holds for a much larger class of operators and provides a different proof of the result obtained in \cite{BIT91}. 
On the other hand, we can additionally describe explicitly the impurities generated in the underlying dynamical system in Proposition~\ref{prop-Defects} and prove in Proposition~\ref{prop:DiscreteSpectrum=q_points} the numerical observations addressed in \cite[rem.~3]{BIT91} - (Q\ref{enu:Q4_OrderingPoints}). 
Specifically, we show that such a spectral defect occurs in each (bounded) spectral gap of $\Sigma(r)$ together with one spectral defect in either $(-\infty,\inf\Sigma_V(r))$ or $(\sup\Sigma_V(r),\infty)$ -- the two unbounded spectral gaps of $\Sigma_V(r)$. 
Therefore, we use the hierarchical structure developed in \cite{Raym95} for $V>4$ (see \cite{BaBeBiRaTh24} for a review) and in \cite{BBL24} for all $V\neq 0$. 
We postpone the formal details and just present here a short example to explain the main ideas.

\begin{exam}
Let $r=\frac{2}{3}$. Then the associated configuration $\omega_r\in\{0,1\}^\ZM$ is given by
$$
\omega_r= \ldots 011 \ 011 \ 011 \cdot \underline{011} \ 011 \ 011 \ldots \,.
$$
Here $\cdot$ marks the origin, namely the first letter to the right of $\cdot$ equals to $\omega_\alpha(0)$. 
The underlined pattern $011$ repeats infinitely often in both directions since $\omega_r$ is $3$-periodic.
The unique ($3$-)Farey neighbors of $r=\frac{2}{3}$ (see Section~\ref{ssect-FareyNumbers} for the definition) are $r_*=\frac{1}{2}$ and $r^*=\frac{1}{1}$. Their associated configurations are given by 
$$
\omega_{r_*}=\omega_{\frac{1}{2}} = \ldots 01 \ 01 \ 01 \cdot\underline{01} \ 01 \ 01 \ldots
\qquad \textrm{and}\qquad
\omega_{r^*}=\omega_{1} = \ldots 1 \ 1 \ 1 \cdot\underline{1} \ 1 \ 1 \ldots \,.
$$
Observe that the corresponding periods (underlined patches) are given by $01$ respectively $1$. 
These two patches determine the impurities (marked with squares) of $\omega_{r_-}$ and $\omega_{r_+}$ with respect to $\omega_r$. 
Specifically, we have 
$$
\omega_{r_+} = \ldots 110 \ 110 \ 110 \ \boxed{1}  \cdot 110 \ 110 \ 110 \ldots
\quad \textrm{and}\quad
\omega_{r_-} = \ldots 101 \ 101 \ 101 \ \boxed{10} \cdot 110 \ 110 \ 110 \ldots \,.
$$
Similarly, the impurities of $\omega_{r_-}$ and $\omega_{r_+}$ for a general rational $r=\frac{p}{q}\in[0,1]$ are uniquely determined by the associated Farey neighbors of $r$, see Proposition~\ref{prop-Defects}.
On the spectral side, we additionally show that for each $V\neq 0$, there are exactly $q=3$ points in the spectral gaps of $\sigma(H_{\frac{2}{3},V})$ as sketched in Figure~\ref{fig-Kohmoto}. 
For $s=\frac{1}{4}=\frac{p'}{q'}$, note also that we have $q'=4$ points in the spectral gaps of $\sigma(H_{\frac{1}{4},V})$.
This is made precise in Proposition~\ref{prop:DiscreteSpectrum=q_points} for all rational rotations.
\end{exam}

The previous discussion hints that when completing $[0,1]$ to the Farey space $\ol{[0,1]}_F$, each rational $r\in[0,1]$ splits up into an upper limit $r_+$, a lower limit $r_-$ in the Farey metric and $r$ itself, see Figure~\ref{fig-FareyTree_FareySpace}. 
With this at hand, we associate to each $x\in\ol{[0,1]}_F$ a sequence $\omega_x\in\{0,1\}^\ZM$ and a corresponding self-adjoint operator $H_{\omega_x,V}$ defined in Equation~\eqref{eq:SchroedOper}.

\begin{theo}
\label{theo-SpeCon}
For all $V\in\RM$, the spectral map
$$\Sigma_V:\overline{[0,1]}_F\to\Kk(\RM)
	\,,\quad x\mapsto\sigma\big(H_{\omega_x,V}\big)\,,
$$
is Lipschitz-continuous. 
In particular, there is a $C=C(V) > 0$ such that 
$$
d_H(\sigma(H_{\omega_x,V}), \sigma(H_{\omega_y,V}))  \leq C \cdot d_F(x,y),
	\qquad x,y \in \ol{[0,1]}_F.
$$
\end{theo}

\begin{proof}
This statement is a special case of Theorem~\ref{theo-SpeCon_PatternEquiv}.
\end{proof}

Note that our proof is not based on $C^\ast$-algebraic techniques since we use \cite{BBC18,BT21}. 
In addition, we obtain explicit spectral estimates leading to the following for Diophantine approximations.

\begin{coro}
\label{cor-LipCon_preview}
For each $V\in\RM$, there is a constant $C=C(V)>0$ such that for all  $\alpha,\frac{p}{q}\in[0,1]$ with  $0<| \alpha-\frac{p}{q} | < \frac{1}{q^2}$,
$$
d_H\big(\sigma(H_{\alpha,V}),\sigma(H_{\frac{p}{q},V})\big) 
	\leq   \frac{C}{q}.
$$
\end{coro}

\begin{proof}
This is a special case of Corollary~\ref{cor-LipCon_PatternEquiv}.
\end{proof}

This result recovers the continuity of $\Sigma_V$ with respect to the Euclidean metric at the irrational points proven in \cite{BIT91} but provides also explicit spectra estimates.
Moreover, Corollary~\ref{cor-LipCon_preview} should be compared with the recent works \cite{BaBePoTe24,Ten24}. 
There the authors prove similar spectral estimates for systems defined by substitutions. 
Note that if $\alpha$ has an eventually periodic continued fraction expansion, then it can be described by a substitution. 

Corollary~\ref{cor-LipCon_preview} indicates a square root behavior of the spectrum if measured with respect to the Euclidean distance. 
Like in the Almost-Mathieu operator, this behavior is in general sharp, see Theorem~\ref{theo-Optimality}. 
The Almost-Mathieu operator is a closely related model with the so-called Hofstadter butterfly. 
The corresponding discussion of sharp spectral estimates for the Almost-Mathieu operator can be found in \cite{BeRa90,HaKaTa16}, confer also a discussion in \cite{BT21}.

Let us shortly explain idea behind Theorem~\ref{theo-Optimality} along an example here and postpone the general case to Section~\ref{sect-SpectralMapContOpt}. 
We prove there that the spectral estimates obtained in Theorem~\ref{theo-SpeCon} are in general optimal, i.e. cannot be improved for all points in $\ol{[0,1]}_F$.
For instance, let $0_+\in\ol{[0,1]}_F$ be the upper limit in the Farey metric of the sequence $(\frac{1}{n})_{n\in\N}$. 
As a consequence of Theorem~\ref{theo-Optimality}, there are constants $C_2>C_1>0$ and a subsequence $(\frac{1}{n_j})_{j\in\N}$ of $(\frac{1}{n})_{n\in\N}$ such that
$$
C_1 d_F\left(0_+,\frac{1}{n_j}\right)
	\leq d_H\big(\sigma(H_{0_+,V}),\sigma(H_{\frac{1}{n_j},V})\big) 
	\leq   C_2 d_F\left(0_+,\frac{1}{n_j}\right)
$$
where $d_F\big(0_+,\frac{1}{n}\big)=\frac{1}{n}$ for all $n\in\N$.

This optimality is obtained when spectral gaps are closing when approaching $0_+$, as indicated in Figure~\ref{fig-Kohmoto} with the red lines. 
Similar results are obtained at each $r\in[0,1]\cap\QM$ and its corresponding limit points $r_-,r_+\in\ol{[0,1]}_F$. 
It is still an open question if the estimates in Theorem~\ref{theo-SpeCon} are also optimal at the irrational points $\alpha\in[0,1]\setminus\QM$, confer a similar discussion for the Almost-Mathieu operator \cite{JiKr19}.

The paper is organized as follows. In Section~\ref{sect-Background} a short introduction to the mechanical words and the Farey numbers is provided. The Farey space is introduced in Section~\ref{sect-RepresentationsFarey}. In addition, its representation as the boundary of a tree and the dynamical representation of the Farey space are established. This is applied to determine the defects from the dynamical perspective in Section~\ref{sect-defects}. Then in Section~\ref{sect-SpectralApplic} we use the previously developed theory to prove Theorem~\ref{theo-SpeCon} and its consequences on the essential spectrum and some structure of the Kohmoto butterfly. Finally, we show that the spectral estimates obtained in Theorem~\ref{theo-SpeCon} are optimal in Section~\ref{sect-SpectralMapContOpt}.

\subsection*{Acknowledgment} S.B. and Y.T. are grateful to Ram Band for various interesting discussions on the Kohmoto butterfly and its presentation in this work. 
S.B. and Y.T. are thankful to Lior Tenenbaum for inspiring discussions on the Lebesgue measure of approximations and for pointing out flaws and typos in an earlier versions.
In addition, S.B. would like to thank Daniel Lenz for an interesting discussion on Sturmian systems and their factorization through the torus. 
This research was supported through the program ``Research in Pairs" by the Mathematisches Forschungsinstitut Oberwolfach in 2018 hosting S.B. and J.B..
This work was partially supported by the Deutsche Forschungsgemeinschaft [BE 6789/1-1 to S.B.].

\section{A short introduction to mechanical words and Farey numbers}
\label{sect-Background}

The basic notions and properties of the space of subshifts, mechanical words and Farey numbers are introduced.


\subsection{The spaces of mechanical systems}
\label{ssect-subshifts}

\noindent The product space $\csp:=\prod_{n\in\ZM} \Aa=\{\omega:\ZM\to\Aa\}$ over the finite alphabet $\Aa:=\{0,1\}$ is compact and metrizable. The group $\ZM$ acts on $\csp$ via the homeomorphism $\tra:\csp\to\csp,\,\tra(\omega)(n):=\omega(n-1),$ which is called the {\em shift}. In view of our applications, we choose the ultrametric $d$ on $\csp$ defined by 
$$
d(\omega,\rho)	
	:= \inf \left\{ \frac{1}{r+1}  \,:\, r\in \N_0 \; \textrm{such that} \; \omega(j)= \rho(j) \textrm{ for all } |j|<r \right\},
	\qquad \omega,\rho\in\Aa^\ZM,
$$

\noindent where $\N_0:=\N\cup\{0\}$. A subset $\Omega\subseteq\csp$ is called {\em ($\tra$-)invariant} if $\tra(\Omega)\subseteq\Omega$.
We denote by  $\Inv$ the set of all closed, $\tra$-invariant and non-empty subsets of $\csp$. 
Elements of $\Inv$ are called {\em subshifts}. Then $\Inv$ can be naturally equipped with the Hausdorff metric $d_\Inv$, which is defined by
$$
d_\Inv(\Omega_1, \Omega_2)
    := \max \left\{ \sup_{\omega\in \Omega_1} \inf_{\rho \in \Omega_2} d(\omega, \rho),\sup_{\rho \in \Omega_2} \inf_{\omega\in \Omega_1} d(\omega, \rho)  \right\}\,,\qquad
	\Omega_1,\Omega_2 \in\Inv.
$$

\noindent The pair $(\Inv, d_\Inv)$ is a compact metric space \cite{Bec16,BBdN18}. 
For every $\omega\in\csp$, the closure of the {\em orbit} $\Orb(\omega):=\{\tra^m(\omega) : m\in\ZM\}\subseteq\csp$ is an element of $\Inv$, i.e. $\ol{\Orb(\omega)} \in \Inv$. 
A subshift $\Omega\in\Inv$ is called {\em topological transitive} if there is a $\omega\in\Omega$ such that $\Omega=\overline{\Orb(\omega)}$. 
Furthermore, $\Omega$ is called {\em minimal}, if $\Omega=\overline{\Orb(\rho)}$ for all $\rho\in\Omega$.

\noindent For $n\in\N$, $\Aa^n$ denotes the set of all maps $v:\{1,\ldots,n\}\to\Aa$ that are identified with a word (of length $n$) $v=v(1) \ldots v(n)$ where $v(i)\in\Aa$. 
Formally the {\em empty word} of length $0$ is denoted \o. 
A word $v\in\Aa^n$ is called a {\em subword} of $\omega\in\csp$ if there is an $m\in\ZM$ such that $\tra^m(\omega)|_{\{1,\ldots,n\}}=v$. 
The empty word  \o ~ is by definition a subword of any $\omega\in\csp$.  
The set $\dic(\omega)$ of all subwords of $\omega$ is called the {\em dictionary} of $\omega$. 
The intersection $\dic(\omega)\cap\Aa^n$ denotes the set of subwords of $\omega$ of length $n\in\N$. The dictionary associated with a subshift $\Omega \in \Inv$ is defined by $\dic(\Omega) := \bigcup_{\omega \in \Omega} \dic(\omega)$. Note that for all $\omega \in \csp$, we have $\dic(\omega) = \dic(\ol{\Orb(\omega)})$.
For this work the following fact is crucial and follows directly from the definition of the metric.

\begin{proposi}[{\cite[cor.~5.5]{BBdN20}}]
\label{prop-TopHau}
For each $\Omega_1, \Omega_2 \in \Inv$, we have 
$$ 
d_\Inv(\Omega_1, \Omega_2) 
	= \min\left\{ \frac{1}{m+1} \,:\, m\in\N_0 \textrm{ such that } \dic(\Omega_1) \cap \Aa^{m} = \dic(\Omega_2) \cap \Aa^{m} \right\}
$$
where $\dic(\Omega_i) \cap \Aa^{0} = \{\text{\o} \}, \, i = 1,2$ holds always. In particular, $(\Omega_n)_{n\in\N}$ converges to $\Omega$ in $\Jj$, if for all $m\in\N$, there is an $n_m\in\N$ such that $\dic(\Omega_n)\cap \Aa^m = \dic(\Omega)\cap\Aa^m$ for all $n\geq n_m$.
\end{proposi}

For $\alpha\in [0,1]$, define the {\em mechanical word} $\omega_\alpha\in\csp$ by
$$\omega_\alpha(n)
	\; := \;\chi_{[1-\alpha,1)}(n\alpha \mod 1)
	\,,\qquad n\in\ZM\,.
$$
If $\alpha\in [0,1]\setminus\QM$, then $\omega_\alpha$ is also called a {\em Sturmian sequence}. 
For $\alpha\in[0,1]$, set $\Omega_\alpha:=\overline{\Orb(\omega_\alpha)}$ and let $\dic_\alpha:=\dic(\Omega_\alpha)=\dic(\omega_\alpha)$ be the associated dictionary. 
For instance, if $\alpha = \frac{\sqrt{5}-1}{2}$, then $\omega_\alpha$ is a Fibonacci sequence, see e.g. \cite{BG13}.  
We define  
$$
\QS:=\{\Omega_\alpha : \alpha\in [0,1]\}\subseteq \Inv
$$ 
and call its closure $\ol{\QS}$ in the compact space $(\Inv,d_\Inv)$ the {\em space of mechanical systems}. 
Note that $\QS$ includes all Sturmian subshifts $\Omega_\alpha$ for $\alpha \in [0,1]\setminus\QM$. 

The study of the metric space $(\Inv, d_\Inv)$ is linked via \cite{BBdN18,BBC18,BT21} to the convergence of the associated operators. 
This link will be discussed  and used in Section~\ref{sect-SpectralApplic}.
It also allows us to focus mainly on the analysis of this metric space $(\Inv, d_\Inv)$.


\subsection{The Farey numbers}
\label{ssect-FareyNumbers}

In the following, we identify a rational number $r\in\QM$ with the unique irreducible fraction $\frac{p}{q}$, i.e. $p$ and $q$ are coprime. 
We use the convention $0=\frac{0}{1}$ and $1=\frac{1}{1}$. 
For $m \in \N$, we call $\frac{p}{q} \in [0,1]\cap \QM$ an {\em $m$-Farey number} if $q\leq m$. The finite set of $m$-Farey numbers is denoted by $F_m$.  For example, we have
$$ 
F_1 = \{0,1\}, \qquad F_2 = \left\{0, \frac12 ,1 \right\}, \qquad F_3 = \left\{0, \frac13, \frac12, \frac23 ,1 \right\}, \qquad
F_4 = \left\{0, \frac14, \frac13, \frac12, \frac23, \frac34 ,1 \right\}.
$$

\begin{defini}
Let $m\in\N$ and $r,s\in F_m$ with $r<s$. Then $r$ and $s$ are called {\em ($m$-)Farey neighbors} or {\em consecutive ($m$-)Farey numbers} if $(r,s)\cap F_m=\emptyset$. For $r\in F_m$, we denote by $r_*,r^*\in F_m$ the unique ($m$-)Farey neighbors satisfying $r_*<r< r^*$ whenever they exist.
We call $r_*$ and $r^*$ the {\em lower} and {\em upper ($m$-)Farey neighbor}  of $r$ respectively.
\end{defini}

In fact, the Farey neighbors exist for all $r\in F_m\setminus\{0,1\}$. 
Only $0,1 \in F_m$ have just one $m$-Farey neighbor, namely $0^* = \frac{1}{m}$ and $1_* = \frac{m-1}{m}$.
Note that the Farey neighbors $r_*$ and $r^*$ depend on $m$.
If for example $r = \frac{1}{2}$, then  $r_* = 0$ and $r^* = 1$ in $F_2$, while $r_* = \frac{1}{3}$ and $r^* = \frac{2}{3}$ in $F_3$. 
If not further elaborated, we always assume $r_*,r^* \in F_m$ whenever $r \in F_m$ was fixed. 
This is why we drop the $m$ dependence in our notation for the Farey neighbours.

Given $\frac{p}{q},\frac{p'}{q'}\in \QM$ reduced fractions, then we define their {\em mediant} by setting
$$
\frac{p}{q} \oplus \frac{p'}{q'} := \frac{p+p'}{q+q'}.
$$
With this we get the following description for Farey neighbours.

\begin{lemma}
\label{lem-FareyNgb}
Let $m\in\N$ and $r\in F_m$.
\begin{enumerate}[(a)]
\item If $r\not\in\{0,1\}$, then $r=r_* \oplus r^*$.
\item Suppose $r\neq 1$ and let $r^* \in F_m$ be the upper $m-$Farey neighbor of $r$.
Then there is a unique $k > m$ such that $(r,r^*)\cap F_{k-1} = \emptyset$ and $(r,r^*)\cap F_{k} = \{r\oplus r^*\}.$ 
In particular, $r \oplus r^* \notin F_j$ holds for $j<k$.
\end{enumerate}
\end{lemma}

\begin{proof}
(a) This is well-known, see e.g. \cite[cor.~2C]{Sc80}. 

(b) Since $r,r^*$ are $m-$Farey neighbors (i.e. $(r,r^*) \cap F_m = \emptyset$) and $(r,r^*) \cap \QM \neq \emptyset$, there must be a unique $k > m$  with $(r,r^*)\cap F_{k-1} = \emptyset$ and $(r,r^*)\cap F_{k} \neq \emptyset$.
Assume  by contradiction that $\sharp (r,r^*)\cap F_{k} \geq 2$. 
Then there are $\frac{a}{k}, \frac{b}{k} \in (r,r^*)\cap F_{k}$ such that $\frac{p}{q} = r < \frac{a}{k} < \frac{b}{k}$ are consecutive $k-$Farey numbers. 
Thus, statement (a) yields
$$ 
\frac{a}{k} \overset{(a)}{=} \frac{p+b}{q+k} = \frac{p+b-a}{q+k} + \frac{a}{q+k}
	\quad\Longleftrightarrow\quad 
\frac{(q+k)a}{qk} = \frac{p+b-a}{q} + \frac{a}{q}
	\quad\Longleftrightarrow\quad 
\frac{a}{k} = \frac{p+b-a}{q}.
$$
As $r=\frac{p}{q}$ was an $m-$Farey number, we conclude $q \leq m < k$. 
Hence, $\frac{a}{k} \in F_q \subseteq F_{k-1}$ follows proving $\frac{a}{k}\in (r,r^*)\cap F_{k-1}$.
This contradicts the choice of $k$. 
\end{proof}

\noindent A direct consequence of Lemma~\ref{lem-FareyNgb} is that if $r,s$ are Farey neighbors in $F_m$, then the next Farey number appearing in between them is exactly $r \oplus s$.  The distance between consecutive $m$-Farey numbers can be estimated from above and below.

\begin{lemma}
\label{lem-FareyDist}
For $m\in\N$, we have $\frac{1}{m^2} \leq \min\limits_{s_1\neq s_2\in F_m} |s_1-s_2|$ and $|r-r^*|\leq \frac{1}{m}$ for $r\in F_m$.
\end{lemma}

\begin{proof}
Let $\frac{p_1}{q_1},\frac{p_2}{q_2}\in F_m$ be two consecutive $m$-Farey numbers. Then $p_2q_1-p_1q_2=1$ follows inductively from Lemma~\ref{lem-FareyNgb}~(a), see also \cite[thm.~2A]{Sc80}. Thus,
$$
\left| \frac{p_1}{q_1}-\frac{p_2}{q_2} \right| 
	= \left| \frac{p_2q_1-p_1q_2}{q_1q_2} \right| 
	= \frac{1}{q_1q_2}
	\geq \frac{1}{m^2}
$$
follows proving the lower bound.

For the upper bound, observe that $\frac{0}{m},\frac{1}{m},\frac{2}{m},\ldots,\frac{m-1}{m},\frac{m-1}{m}$ are all elements in $F_m$, while they are not necessarily reduced fractions. Thus, if $r\in F_m$, then there is a $0\leq j\leq m-1$ such that $\frac{j}{m}\leq r<r^* \leq \frac{j+1}{m}$. Hence, $|r-r^*|\leq \frac{1}{m}$ is concluded.
\end{proof}

\section{Various representations of the mechanical systems}
\label{sect-RepresentationsFarey}

This section is devoted to study the space of mechanical systems $\ol{\QS}$ as a metric space and different isometric representations.
One is obtained by defining a new metric -- called the {\em Farey-metric $d_F$} -- on the interval $[0,1]$ inherited from the Farey number.
Using results from \cite{Mi91} -- linking Farey numbers and the dictionaries of Sturmian words -- we show that the map 
$$
[0,1] \to \QS, \alpha \mapsto \Omega_\alpha
$$
is isometric if $[0,1]$ is equipped with the Farey metric $d_F$, see Theorem~\ref{theo-Homeo_Farey-Subshifts}. Then this map uniquely extends to an isometric surjective map $\ol{[0,1]}_F\to\ol{\QS}$.
In addition, we give another representation of these spaces as the boundary of an infinite tree in Section \ref{ssect-FareyGraphRep}.


\subsection{The Farey space}
\label{ssect-FareyTopoRep}


The {\em Farey metric $d_F:[0,1]\times[0,1]\to[0,\infty)$} on the interval $[0,1]$ is defined by
\[
d_F(\alpha, \beta) := 
	\begin{cases}
		0, & \alpha = \beta,\\ 
		1, & \alpha \in \{0,1\} \text{ or } \beta \in \{0, 1\},\\
		\min\big\{\frac{1}{m+1}: \exists r \in F_m \text{ with } \alpha, \beta \in (r, r^*)\big\}, & \text{else.} 
	\end{cases}
\]
We will see that this defines an ultrametric on $[0,1]$, see Proposition~\ref{prop-dF_metric} below. The set
$$
B_F(\alpha,R) := \{\beta\in[0,1] \,:\, d_F(\alpha,\beta)<R\}
$$
denotes the open ball around $\alpha\in[0,1]$ of radius $R>0$ in the Farey metric.

\begin{proposi}
\label{prop-dF_metric}
The Farey metric $d_F$ defines an ultrametric on $[0,1]$, i.e. 
$$
d_F(\alpha,\beta)\leq \max\big\{ d_F(\alpha,\gamma),\, d_F(\gamma,\beta)\big\},
	\qquad
	\alpha,\beta,\gamma\in[0,1].
$$
Moreover, the following holds for $\alpha,\beta\in[0,1]$ and $m\in\N$.
\begin{enumerate}[(a)]
\item \label{enu:dF=1} If $d_F(\alpha,\beta)=1$, then $\alpha\in\{0,1\}$ or  $\beta\in\{0,1\}$.
\item \label{enu:dF_lowerBound} If $\alpha\neq\beta$ and $\alpha\leq r\leq\beta$ holds for some $r\in F_m$, then $d_F(\alpha,\beta)> \frac{1}{m+1}$.
\item \label{enu:dF_Ball_rat} If $\alpha\in F_m$, then $B_F\big(\alpha,\frac{1}{m+1}\big)=\{\alpha\}$.
\item \label{enu:dF_Ball_irrat} If $\alpha\not\in F_m$, then there exists an $r\in F_m$ such that $B_F\big(\alpha,\frac{1}{m}\big)=(r,r^*)$.
\item \label{enu:dF_Eucl} We have $|\alpha-\beta|\leq 2d_F(\alpha,\beta)$.
\end{enumerate}
\end{proposi}

\begin{proof}
First note that for $\alpha,\beta\in[0,1]$,
\begin{equation}
\label{eq:FareyMetric_range}
d_F(\alpha,\beta) \in \{0\}\cup \left\{ \frac{1}{m} \,:\, m\in\N \right\}.
\end{equation}
We first prove \eqref{enu:dF=1} to \eqref{enu:dF_Eucl}. Statement \eqref{enu:dF=1} follows immediately from the definition. For \eqref{enu:dF_lowerBound}, suppose $\alpha\neq\beta$ and $\alpha\leq r\leq\beta$ holds for some $r\in F_m$. Thus, there exists no $s\in F_m$ such that $\alpha,\beta\in (s,s^*)$ proving $d_F(\alpha,\beta)> \frac{1}{m+1}$.

For \eqref{enu:dF_Ball_rat}, suppose $\alpha\in F_m$. If $\beta\in[0,1]\setminus\{\alpha\}$, then \eqref{enu:dF_lowerBound} implies $d_F(\alpha,\beta)> \frac{1}{m+1}$ proving $B_F\big(\alpha,\frac{1}{m+1}\big)=\{\alpha\}$. 

For \eqref{enu:dF_Ball_irrat}, suppose $\alpha\not\in F_m$. Then there is a unique $r\in F_m$ such that $\alpha\in(r,r^*)$. If $\beta\in (r,r^*)$, then $d_F(\alpha,\beta)\leq \frac{1}{m+1}<\frac{1}{m}$. If $\beta\in[0,1]\setminus(r,r^*)$, then \eqref{enu:dF_lowerBound} leads to $d_F(\alpha,\beta)> \frac{1}{m+1}$. Thus, $B_F\big(\alpha,\frac{1}{m}\big)=(r,r^*)$ follows using Equation~\eqref{eq:FareyMetric_range}.

For \eqref{enu:dF_Eucl}, note that the statement trivially holds if $\alpha\in\{0,1\}$ or $\beta\in\{0,1\}$. Let $\alpha,\beta\in(0,1)$ be different. By Equation~\eqref{eq:FareyMetric_range} and \eqref{enu:dF=1}, there is an $m\in\N$ such that $d_F(\alpha,\beta)=\frac{1}{m+1}$. Thus, there is an $r\in F_m$ satisfying $\alpha,\beta\in (r,r^*)$. Then $|r-r^*|\leq \frac{1}{m}$ holds by Lemma~\ref{lem-FareyDist} implying 
$$
|\alpha-\beta|\leq\frac{1}{m}\leq \frac{2}{m+1} = 2d_F(\alpha,\beta).
$$

Finally, we show that $d_F$ is an ultrametric. Clearly, $d_F$ is symmetric and definite and it suffices to show $d_F(\alpha,\beta)\leq \max\big\{ d_F(\alpha,\gamma),\, d_F(\gamma,\beta)\big\}$ for all $\alpha,\beta,\gamma\in[0,1]$. If $d_F(\alpha,\beta)=1$, then (a) yields $\alpha\in\{0,1\}$ or  $\beta\in\{0,1\}$. Thus, $\max\big\{ d_F(\alpha,\gamma),\, d_F(\gamma,\beta)\big\}=1$ is concluded.

Now suppose $d_F(\alpha,\beta)<1$. Then there is an $m\in\N$ and an $r\in F_m$ such that $d_F(\alpha,\beta)=\frac{1}{m+1}$ and $\alpha,\beta\in(r,r^*)$. If $\gamma\not\in(r,r^*)$, then $d_F(\alpha,\gamma)> \frac{1}{m+1}=d_F(\alpha,\beta)$ follows from (b).
Next suppose $\gamma\in(r,r^*)$. Since $d_F(\alpha,\beta)=\frac{1}{m+1}$, there is an $s\in F_{m+1}$ such that $\alpha\leq s \leq \beta$. Thus, either $\gamma<s$ or $\gamma\geq s$. If $\gamma<s$, then \eqref{enu:dF_lowerBound} yields $d_F(\gamma,\beta)>\frac{1}{m+2}$. Hence, Equation~\eqref{eq:FareyMetric_range} leads to $d_F(\gamma,\beta)\geq \frac{1}{m+1}=d_F(\alpha,\beta)$ proving the claim. If $\gamma\geq s$, one concludes similarly $d_F(\alpha,\gamma)\geq \frac{1}{m+1}=d_F(\alpha,\beta)$ finishing the proof.
\end{proof}

It is not difficult to check that $[0,1]$ equipped with $d_F$ is not a complete metric space. For instance $(\frac1n)_{n\in\N}$ is a Cauchy sequence but it admits no limit with respect to $d_F$. Denote by $[0,1]_F$ the interval $[0,1]$ equipped with the topology inherited from $d_F$. Furthermore, $\ol{[0,1]}_F$ denotes the corresponding completion with respect to $d_F$.

\begin{defini}
The topology on $\ol{[0,1]}_F$ induced by $d_F$ is called {\em Farey topology} and the metric space $(\ol{[0,1]}_F,d_F)$ is called {\em Farey space}.
\end{defini}

\begin{rem}
A basis for the Farey topology on $[0,1]$ is given by 
$$
\bs=\big\{ [r,s] \,:\, r,s\in\QM\cap[0,1] \big\}.
$$
Note that Proposition~\ref{prop-dF_metric}~\eqref{enu:dF_Eucl} asserts that the Farey topology is finer than the Euclidean topology, which can be deduced also from the bases of the topologies.
In particular, convergence in $[0,1]_F$ implies convergence in $[0,1]$. 
In addition, we point out that while $\ol{[0,1]}_F$ is a totally disconnected space, it has plenty isolated points, which are all the rational numbers in $[0,1]$.
\end{rem}


Since we are dealing with different topologies on $[0,1]$, we denote the Euclidean limit of $(x_n)_{n \in \N} \subseteq [0,1]$ by $\lim_E x_n$ and the limit in the Farey-topology by $\lim_F x_n$ (provided they exist).

\begin{proposi}
\label{prop-convergence_in_farey_topo}
Let $(\alpha_n)_{n\in \N}\subseteq [0,1]_F$ be a Cauchy sequence with respect to $d_F$.  
Then exactly one of the following assertions hold:
\begin{enumerate}[(a)]
\item The limit $\alpha:=\lim_E \alpha_n\in[0,1]$ is irrational and $\alpha=\lim_F \alpha_n$.
\item The limit $\alpha:=\lim_E \alpha_n\in[0,1]$ is rational and
	\begin{enumerate}[(b1)]
		\item \label{enu:ComplFarey_rational=} $\alpha_n = \alpha$ for $n$ large enough;
		\item \label{enu:ComplFarey_rational>}$\alpha_n > \alpha$ for $n$ large enough;
		\item \label{enu:ComplFarey_rational<}$\alpha_n < \alpha$ for $n$ large enough.
	\end{enumerate}
\end{enumerate}
\end{proposi}

\begin{proof}
According to Proposition~\ref{prop-dF_metric}~\eqref{enu:dF_Eucl}, the limit $\alpha:=\lim_E \alpha_n\in[0,1]$ exists as $(\alpha_n)_{n\in \N}$ is Cauchy with respect to $d_F$. 
If $\alpha$ is irrational, then for each $m\in\N$, there is a unique $r\in F_m$ satisfying $\alpha\in(r,r^*)$. Thus, Proposition~\ref{prop-dF_metric}~\eqref{enu:dF_Ball_irrat} implies that for each $m\in\N$, there is an $n_m\in\N$ such that $\alpha_n\in B_F(\alpha,\frac{1}{m})=(r,r^*)$ for $n\geq n_m$. Hence, $\alpha=\lim_F \alpha_n$ follows.

Next suppose $\alpha=\frac{p}{q}$ with $p$ and $q$ coprime. Assume towards contradiction that neither \eqref{enu:ComplFarey_rational=},\eqref{enu:ComplFarey_rational>} nor \eqref{enu:ComplFarey_rational<} holds. Thus, for each $N\in\N$, there are $n_1,n_2\geq N$ such that $\alpha_{n_1}<\frac{p}{q}\leq\alpha_{n_2}$ or $\alpha_{n_1}\leq \frac{p}{q}<\alpha_{n_2}$. Both cases imply by Proposition~\ref{prop-dF_metric}~\eqref{enu:dF_lowerBound} that $d_F(\alpha_{n_1},\alpha_{n_2})>\frac{1}{q+1}$. This contradicts that $(\alpha_n)_{n\in \N}$ is Cauchy with respect to $d_F$. 
\end{proof}

As a direct consequence, the Farey space can be represented as the interval $[0,1]$ where 
\begin{itemize}
\item each rational point $r\in(0,1)\cap\QM$ is split up into three points $r_-,r_+$ and $r$, and
\item $0$ is split up into two points $0$ and $0_+$, and
\item $1$ is split up into two points $1$ and $1_-$.
\end{itemize}
This is formulated in the following corollary and skteched in Figure~\ref{fig-FareyTree_FareySpace}.

\begin{coro}
\label{cor:Completion_FareySpace}
For $r\in\QM\cap[0,1]$, let $r_+:=\lim_F r+\frac{1}{n}$ and $r_-=\lim_F r-\frac{1}{n}$. Then
$$
\ol{[0,1]}_F = [0,1]\cup \{r_+ \,:\, r \in \QM \cap[0,1)\} \cup  \{r_- \,:\, r \in \QM \cap (0,1]\}.
$$
\end{coro}

\begin{proof}
This is an immediate consequence of Proposition~\ref{prop-convergence_in_farey_topo}.
\end{proof}

We emphasize that the definition of $r_+$ and $r_-$ is independent of the specific sequence $\alpha_n:=r\pm \frac{1}{n}$. 



\begin{rem}
We note that for irrational $\alpha\in[0,1]$, the associated subshift $\Omega_\alpha$ can be factorized through a circle $\RM / \ZM$ where the circle splits at each $n\alpha \mod 1, n\in\ZM,$ into two points, see e.g. \cite[prop.~5.1]{DaLe04}. 
This cutting is different to the representation of the subshifts $\ol{\QS}$ developed here via $\ol{[0,1]}_F$ where the interval splits at the rational points into three points.
\end{rem}

\subsection{The dynamical representation}
\label{ssect-FarexDynRep}

In this section, we show that the Farey space $\ol{[0,1]}_F$ and the mechanical dynamical systems $\ol{\QS}$ equipped with $d_\Inv$ are isometrically isomorphic. 

For a set $A$, let $\sharp A$ be its cardinality and define the {\em complexity function $\scf$} by
$$
\scf:\N\times\Inv\to\N, \quad 
	\scf(n,\Omega) := \sharp \dic(\Omega)\cap\Aa^n.
$$
The following statement is well-known, see e.g. \cite{Fog02,BG13}.

\begin{lemma}
\label{lem:Complexity}
Let $\alpha\in[0,1]$. Then the following assertions hold.
\begin{enumerate}[(a)]

\item Let $m\in \N$ and $\alpha\in[0,1]\setminus F_m$ , then $\scf(m,\Omega_\alpha)=m+1$. 
In particular, if $\alpha\in[0,1]\setminus\QM$, then $\scf(m,\Omega_\alpha)=m+1$ holds for all $m\in\N$.

\item If $\alpha=\frac{p}{q}\in[0,1]$ with $p$ and $q$ coprime, then $\scf(m,\Omega_\alpha)=q$ for $m\geq q$ and $\scf(j,\Omega_\alpha)=j+1$ for $1\leq j\leq q-1$.
\end{enumerate}
\end{lemma}

The key connection of the mechanical systems and Farey numbers are the following results. Those are consequences from a work by Mignosi \cite{Mi91}, see also the recent elaboration on that \cite{Th23}.

\begin{lemma}
\label{lem:DifferentDictionaries}
Let $m \in \N$ and $\alpha,\beta \in F_m$ be different, then $\dic_\alpha \cap \Aa^m \neq \dic_\beta \cap \Aa^m$. 
\end{lemma}

\begin{proof}
Without loss of generality, suppose $\alpha < \beta$. For the induction base $m=1$, observe that 
\[
F_1=\{0,1\} 
\qquad\textrm{and}\qquad
\dic_0 \cap \Aa^1 = \{0\} \neq  \{1\} = \dic_0 \cap \Aa^1.
\]
For the induction step, suppose that for all different $\alpha,\beta \in F_m$, we have $\dic_\alpha \cap \Aa^m \neq \dic_\beta \cap \Aa^m$. 
Note that $F_j\subseteq F_m$ for $j\leq m$ and $\dic_\alpha \cap \Aa^m \neq \dic_\beta \cap \Aa^m$ implies $\dic_\alpha \cap \Aa^n \neq \dic_\beta \cap \Aa^n$ for $n\geq m$.

Let $\alpha = \frac{p}{q},\beta = \frac{p'}{q'}\in F_{m+1}$ be different and irreducible. 
If $q\neq q'$, then Lemma~\ref{lem:Complexity}~(b) implies $\dic_\alpha \cap \Aa^n \neq \dic_\beta \cap \Aa^n$ for $n=\max\{q,q'\}\leq m+1$ since their cardinalities differ. 
Thus, there is no loss of generality in assuming $q=q'$. 
If $q=q'<m+1$, then $\alpha,\beta\in F_m$ and so $\dic_\alpha \cap \Aa^{m+1} \neq \dic_\beta \cap \Aa^{m+1}$ follows from the induction base. 

It is left to treat the case $q=q'=m+1$. 
Since $m+1\geq 2$, we have $\alpha>0$.
Assume towards contradiction that $\dic_\alpha \cap \Aa^{m+1} = \dic_\beta \cap \Aa^{m+1}$. 
Lemma~\ref{lem-FareyNgb} asserts that there exist $\beta_*,\beta^*\in F_j$ for some $j<m+1$ such that $\beta=\beta_* \oplus \beta^*$. 
Since $j<m+1$, we conclude (by Lemma~\ref{lem:Complexity}) $\alpha\neq \beta_*$ and $\scf(m+1,\Omega_{\beta_*})\leq j < m+1$. 
Let $\gamma \in (\beta_*, \beta)$. 
Then $\dic_\beta \cap \Aa^{m+1} \subset \dic_\gamma \cap \Aa^{m+1}$ holds by \cite[cor.~9]{Mi91}. 
With this at hand, \cite[thm.~16]{Mi91} applied to $0<\alpha<\beta_*<\gamma<\beta$ leads to
\[ 
(\dic_\beta \cap \Aa^{m+1})\setminus \dic_{\beta_*} \cap \dic_\alpha \subseteq (\dic_\gamma \cap \Aa^{m+1})\setminus \dic_{\beta_*} \cap \dic_\alpha= \emptyset.
\]
Thus, our assumption $\dic_\alpha \cap \Aa^{m+1} = \dic_\beta \cap \Aa^{m+1}$ leads to $(\dic_\beta \cap \Aa^{m+1})\setminus \dic_{\beta_*} = \emptyset$. 
Hence, $\dic_{\beta} \cap \Aa^{m+1} \subseteq \dic_{\beta_*} \cap \Aa^{m+1}$ follows. 
Since $\beta= \frac{p'}{q'}$ is irreducible, Lemma~\ref{lem:Complexity}~(b) yields
\[
m+1 = \scf(m+1, \Omega_\beta) \leq \scf(m+1, \Omega_{\beta_*}) \leq j < m+1,
\]
a contradiction.
Therefore $\alpha,\beta\in F_{m+1}$ different implies $\dic_\alpha \cap \Aa^{m+1} \neq \dic_\beta \cap \Aa^{m+1}$.
\end{proof}

\begin{proposi}
\label{prop-FarDic}
Let $\alpha,\beta\in[0,1]$ and $m\in\N$. Then $\dic_\alpha\cap\Aa^m=\dic_\beta\cap\Aa^m$ holds if and only if there exists an $r\in F_m$ with $\alpha,\beta\in(r,r^*)$.
%
%
%
\end{proposi}

\begin{proof}
If $\alpha,\beta\in(r,r^*)$, then the equation $\dic_\alpha\cap\Aa^m=\dic_\beta\cap\Aa^m$ follows from \cite[cor.~10]{Mi91}. 
For the converse direction assume without loss of generality $\alpha<\beta$. 
By contraposition, we prove that if there is an $r\in F_m$ with $\alpha\leq r\leq \beta$, then $\dic_\alpha\cap\Aa^m\neq\dic_\beta\cap\Aa^m$.  
For $\alpha\in F_m\setminus\{r\}$ and $\beta=r$, we get $\dic_\alpha \cap \Aa^m \neq \dic_\beta \cap \Aa^m$ directly from Lemma~\ref{lem:DifferentDictionaries}.
If $\alpha<r=\beta$ and $\alpha\not\in F_m$, then 
$$
\scf(m,\Omega_\alpha)=m+1>m=\scf(m,\Omega_\beta)
$$
follows from Lemma~\ref{lem:Complexity} proving  $\dic_\alpha\cap\Aa^m\neq\dic_\beta\cap\Aa^m$.
If $\alpha=r<\beta$, then $\dic_\alpha\cap\Aa^m\neq\dic_\beta\cap\Aa^m$ is similarly concluded.

The last case to treat is when $\alpha$ nor $\beta$ are in $F_m$.
Since there is still $r \in F_m$ with $\alpha < r < \beta$, we can choose $r$ such that $\beta \in (r,r^*)$. 
For this choice \cite[thm.~16]{Mi91} applies and yields
$$
\emptyset = \big((\dic_\beta\cap\Aa^m) \setminus \dic_{r}\big)\cap\dic_\alpha 
	= \big((\dic_\beta\cap\Aa^m) \setminus \dic_{r}\big)\cap(\dic_\alpha\cap\Aa^m).
$$
Assume by contradiction that $\dic_\alpha\cap\Aa^m=\dic_\beta\cap\Aa^m$. Then the previous considerations imply $\dic_\alpha\cap\Aa^m\subseteq \dic_r\cap\Aa^m$. On the other hand, 
$$
\scf(m,\Omega_\alpha)=m+1>m=\scf(m,\Omega_r)
$$
holds by Lemma~\ref{lem:Complexity}, a contradiction.
\end{proof}

The latter statement is the curcial ingredient to prove that the following map is isometric. 

\begin{theo}
\label{theo-Homeo_Farey-Subshifts}
There is a unique surjective isometry $\Phi:(\ol{[0,1]}_F,d_F) \to (\ol{\QS},d_{\Inv})$ such that $\Phi(\alpha)=\Omega_\alpha$ for $\alpha\in[0,1]$. In particular, $d_F(x,y)=d_\Inv(\Phi(x),\Phi(y))$ for all $x,y\in\ol{[0,1]}_F$.
\end{theo}

\begin{proof}
First note that it suffices to show that the map $[0,1]_F\to\QS,\alpha\mapsto\Omega_\alpha,$ is bijective and isometric. By definition of $\QS$, this map is surjective. Hence, it is sufficient to prove $d_F(\alpha,\beta)=d_\Inv(\Omega_\alpha,\Omega_\beta)$ for all different $\alpha,\beta\in [0,1]_F$. 

If $\alpha\in\{0,1\}$ or $\beta\in\{0,1\}$, then $d_F(\alpha,\beta)=1$. Since  $\alpha\neq \beta$, $\dic_1\cap\Aa^1=\{1\}$ and $\dic_0\cap\Aa^1=\{0\}$, we conclude $d_\Inv(\Omega_\alpha,\Omega_\beta)=1=d_F(\alpha,\beta)$.

Next, let $\alpha,\beta\in(0,1)$ and $m\in\N$ be such that $\alpha<\beta$ and $d_F(\alpha,\beta)=\frac{1}{m+1}$. Hence, there is an $r\in F_m$ and an $s\in F_{m+1}$ such that $\alpha,\beta\in (r,r^*)$ and $\alpha\leq s\leq \beta$. Thus, Proposition~\ref{prop-FarDic} leads to 
$$
\dic_\alpha\cap\Aa^m = \dic_\beta\cap\Aa^m 
	\qquad\textrm{and}\qquad
\dic_\alpha\cap\Aa^{m+1} \neq \dic_\beta\cap\Aa^{m+1}.
$$
Hence, $d_\Inv(\Omega_\alpha,\Omega_\beta)=\frac{1}{m+1}$ is concluded from Proposition~\ref{prop-TopHau}, proving $d_\Inv(\Omega_\alpha,\Omega_\beta)=d_F(\alpha,\beta)$.
\end{proof}

\begin{coro}
\label{cor-FareyCpt}
The Farey space $\ol{[0,1]}_F$ is compact.
\end{coro}

\begin{proof}
The space $\Inv$ is compact \cite{Bec16,BBdN18}. Therefore $\ol{\QS} \subseteq \Inv$ is compact in the induced topology. The statement then follows from Theorem~\ref{theo-Homeo_Farey-Subshifts}, as $\Phi$ is a homeomorphism. 
\end{proof}

\begin{coro}
\label{cor:complexity}
Let $\Phi:(\ol{[0,1]}_F,d_F) \to (\ol{\QS},d_{\Inv})$ be the surjective isometry from Theorem~\ref{theo-Homeo_Farey-Subshifts}.
\begin{enumerate}[(a)]
\item If $x\in\ol{[0,1]}_F\setminus([0,1]\cap\QM)$, then $\scf(n,\Phi(x))=n+1$ for all $n\in\N$.
\item If $x\in [0,1]\cap\QM$, then there exists an $n_x\in\N$ such that $\scf(n,\Phi(x))=n$ for all $n\geq n_x$.
\end{enumerate}
\end{coro}

\begin{proof}
If $x\in[0,1]$, then the statement follows from Lemma~\ref{lem:Complexity}. Let $x\in\ol{[0,1]}_F\setminus[0,1]$ and $n\in\N$. Then there is a sequence $(x_k)_{k\in\N}\subseteq[0,1]\setminus\QM$ such that $\lim_F x_k=x$. Since $x_k$ is irrational, we have $\scf(n,\Phi(x_k))=n+1$. Then Proposition~\ref{prop-TopHau} implies that the complexity function is continuous in $\Inv$, namely $\scf(n,\Phi(x))=\lim_{k\to\infty}\scf(n,\Phi(x_k))=n+1$.
\end{proof}


\subsection{The graph representation}
\label{ssect-FareyGraphRep}
In this section, we isometrically represent the Farey space as the boundary of a tree which is closely related to the Farey tree.

A tuple $\Gg = (\Vv, \Ee)$ is called a {\em directed graph} where $\Vv$ is a countable set and $\Ee\subseteq \Vv\times\Vv$. 
We call $\Vv$ the {\em vertex set} and $\Ee$ the  {\em edge set} of $\Gg$. 
Then $u\in\Vv$ {\em is connected to} $v\in\Vv$ if $(u,v)\in\Ee$. 
A {\em finite path} is a tuple $(v_0, \dots v_n) \in \Vv^{n+1}$ if $(v_{i-1},v_i)\in\Ee$ for $1\leq i\leq n$. 
If $v_0=v_n$ and $n\in\N$, we call the path a {\em cycle}. 
Finite paths define a (possibly strict, i.e. irreflexive) partial order relation on the vertex set $\Vv$, namely $u\to v$ holds for $u,v\in \Vv$, if there is a finite path $(v_0,\ldots,v_n)$ with $v_0=u$ and $v_n=v$. 
A directed graph $\Gg$ is called a {\em rooted directed tree}, if $\Gg$ has no cycles and there is a vertex $o\in\Vv$ such that for each $v\in\Vv\setminus\{o\}$, we have $o\to v$. 
Note that this vertex $o$ is unique (as the graph admits no cycles) and it is called the {\em root}.

We recursively define a rooted directed tree $\Tt_F=(\Vv_F,\Ee_F)$, where each vertex gets a label assigned, namely, we have a map $L:\Vv_F\to\Ll$ where
$$
\Ll:= \big\{ \{r\} \,:\, r\in\QM\cap[0,1]\big\} \cup \bigcup_{m\in\N} \big\{ (r,r^*) \,:\, r\in F_m\big\} \cup \{[0,1]\}.
$$
For convenience, the building blocks for this recursive definition are plotted in Figure~\ref{fig-Building_FareyTree}.
The root $o\in \Vv_F$ gets the label $[0,1]\in\Ll$ and it belongs to level $0$. The root $o$ is connected to exactly three vertices $u_1,u_2,u_3$ in level $1$, i.e. $(o,u_i)\in\Ee_F$, such that $L(u_1):=\{0\}$, $L(u_2):=(0,1)$ and $L(u_3):=\{1\}$. We continue recursively defining the tree $\Tt_F$. Let $u\in \Vv_F$ be in level $k\geq 1$.
\begin{itemize}
\item If $L(u)=(r,r^*)\in\Ll$ for some $r\in F_m$, set $s:=r\oplus r^*\in F_n$ where $n>m$ by Lemma~\ref{lem-FareyNgb}.
Then $u$ is connected to exactly three vertices $v_1,v_2,v_3$ in level $k+1$, i.e. $(u,v_i)\in\Ee_F$, such that $L(v_1):=(r,s)\in\Ll$, $L(v_2):=\{s\}\in\Ll$ and $L(v_3):=(s,r^*)\in\Ll$.
\item If $L(u)=\{r\}\in\Ll$ for some $r\in F_m$, then there is exactly one edge $(u,v)\in\Ee_F$ where $v\in \Vv_F$ is in level $k+1$ and $L(v):=\{r\}$.
\end{itemize}

\begin{figure}[htb]
 \centering
	\includegraphics[scale=0.8]{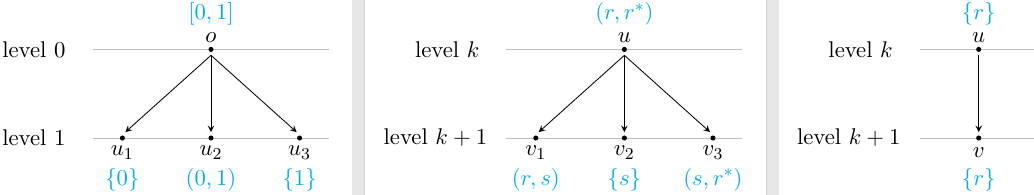}
 \caption{
 The building blocks of the interval-Farey tree are plotted.
 }
 \label{fig-Building_FareyTree}
\end{figure}

\begin{defini}
The previously defined rooted directed tree $\Tt_F=(\Vv_F,\Ee_F)$ is called {\em interval-Farey tree} and $L:\Vv_F\to\Ll$ is called the {\em label map}.
\end{defini}

Note that the label map is surjective but not injective, see Figure~\ref{fig-FareyTree_FareySpace}. 
We emphasize that the interval-Farey tree is closely related to the classical Farey tree. 
The Farey tree was used for the Kohmoto model in \cite{OK85}. 
There the authors construct a renormalization group formulation using the Farey address defined by the Farey tree and not the continued fraction expansion.

\begin{figure}[hb]
 \centering
 	\includegraphics[scale=0.8]{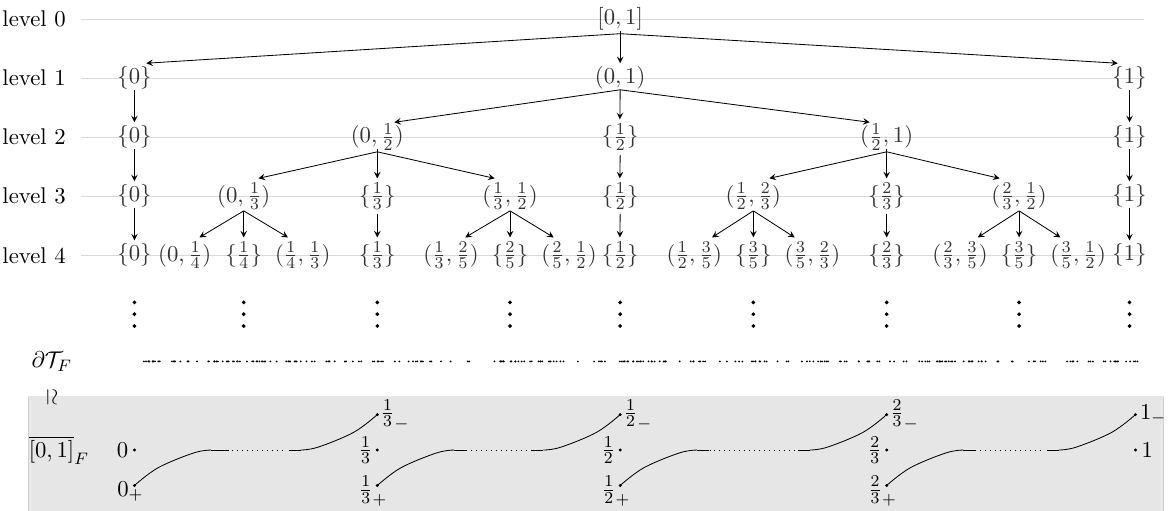}
 \caption{The first levels of the Farey tree $\Tt_F$, its boundary $\partial \Tt_F$ and (in the gray box) the Farey space $\ol{[0,1]}_F$ are sketched.
 Each rational point is splits up into three points (except $0$ and $1$), confer Proposition~\ref{prop-convergence_in_farey_topo}.
 }
 \label{fig-FareyTree_FareySpace}
\end{figure}

All vertices in level $k$ have graph combinatorial distance $k$ to the root $o$, namely these are the vertices in the sphere of radius $k$ around $o$. For convenience, we define the $k$-sphere by
$$
\Vv_F(k) := \{v\in \Vv_F \,:\, v \textrm{ is in level } k\},\qquad k\in\N_0.
$$

\begin{lemma}
\label{lem-BasicFareyTree}
Let $\Tt_F= (\Vv_F,\Ee_F)$ be the interval-Farey tree.
\begin{enumerate}[(a)]
\item If $u\to v$ for $u,v\in\Vv_F$, then $L(v)\subseteq L(u)$.
\item We have $[0,1]=\bigsqcup_{v\in \Vv_F(k)} L(v)$ for all $k \in \N_0$.
\end{enumerate}
\end{lemma}

\begin{proof}
Both assertion follow inductively from the definition of the interval-Farey tree.
\end{proof}

Recall that we seek to isometrically identify the Farey space with the boundary of the interval-Farey tree.
Here the boundary of the the tree $\Tt_F$ is defined by
$$
\partial\Tt_F
	:= \big\{ \gamma:\N_0\to \Vv_F \,:\, \big(\gamma(i-1),\gamma(i)\big)\in\Ee_F \textrm{ and } \gamma(0)=o \big\}.
$$
Then the boundary $\partial\Tt_F$ is a compact space if equipped with the product topology. 
We turn this compact space into a metric space following the Michon-correspondence \cite{Mi85, PB09}. 
A map $\kappa: \Vv_F \to (0,\infty)$, such that $v\to w$ implies $\kappa(w) < \kappa (v)$ and $\lim_{n \to \infty}\kappa(\gamma(n))=0$ for all $\gamma \in \partial\Tt$, is called a {\em weight} on $\Tt_F$.
Moreover, given $\gamma, \eta \in \partial \Tt_F$ different, we define $\gamma \land \eta:=\gamma(n)$, where $n\in\N_0$ is unique number such that $\gamma(n)=\eta(n)$ and $\gamma(n+1)\neq \eta(n+1)$. 
Note that either $L(\gamma(n))=[0,1]$ or $L(\gamma(n))=(r,r^*)$ must hold for some $r\in F_m$ with $m\in\N$. 
This is because if $L(\gamma(n))=\{r\}$, then $\gamma(m)=\eta(m)$ follows for all $m\geq n$ from the construction of the interval-Farey tree.

For the interval-Farey tree, define $\kappa_F: \Vv_F \to [0,\infty)$ by
$$
\kappa_F(v) 
	:= \begin{cases}
		1,\qquad &v=o,\\
		\frac{1}{m+1},\qquad &L(v)=(r,r^*), r\in F_m \textrm{ and } r\oplus r^*\in F_{m+1}\setminus F_m,\\
		\frac{1}{2}\kappa_F(u),\qquad &L(v)=\{r\} \textrm{ and } u\in \Vv_F \textrm{ satisfies } (u,v)\in\Ee_F.
	\end{cases}
$$
Note that $m\in\N$ in the previous definition is chosen to be the largest integer so that $r$ and $r^*$ are Farey neighbors in $F_m$ but not in $F_{m+1}$, compare with Lemma~\ref{lem-FareyNgb}~(b).

\begin{exam}
If the label of $v\in\Vv_F$ is $(\frac{1}{3},\frac{1}{2})$, then $\frac{1}{3}\oplus\frac{1}{2}=\frac{2}{5}$ and so $\kappa_F(v)=\frac{1}{5}$.
If $v\in\Vv_F$ has label $(0,\frac{1}{7})$, then $0\oplus\frac{1}{7}=\frac{1}{8}$ and so $\kappa_F(v)=\frac{1}{8}$.
Also if $(u,v,w)$ is the path labelled by $L(u) =(0,1), L(v) = L(w)= \{\frac{1}{2}\}$, then we have $\kappa_F(u)= \frac{1}{2},\kappa_F(v) = \frac{1}{4}$ and $\kappa_F(w) = \frac{1}{8}$.
\end{exam}

\begin{proposi}
\label{prop-BoundaryMetric}
The map $\kappa_F: \Vv_F \to [0,\infty)$ is a weight on the interval-Farey tree $\Tt_F$ such that $\kappa_F(v)\leq \frac{1}{k+1}$ for all $v\in\Vv_F(k)$ and $k\in\N$. The map $d_{\Tt_F}: \partial\Tt_F \times \partial\Tt_F \to [0,\infty)$ defined by
$$
d_{\Tt_F}(\gamma, \eta) 
	:= \begin{cases}
		0,\qquad &\gamma=\eta,\\
		\kappa_F(\gamma \land \eta),\qquad &\gamma\neq \eta\\
	\end{cases}
$$ 
is an ultra metric inducing the product topology and $(\partial\Tt_F,d_{\Tt_F})$ is a compact metric space.
\end{proposi}

\begin{proof}

We prove inductively over the level $k\in\N$ that $\kappa_F(v)\leq \frac{1}{k+1}$ for all $v\in\Vv_F(k)$ and if $(u,v)\in\Ee_F$, then $\kappa_F(u)<\kappa_F(v)$.
For the induction base $k\in\{0,1\}$, observe that $\kappa_F(o)=1$ and $\kappa_F(u)=\frac{1}{2}$ for $u\in\Vv_F(1)$.
Now suppose the statement is true for $k\in\N$. 
Let $v\in \Vv_F(k+1)$. 
By construction of the interval-Farey tree, there is a unique $u\in \Vv_F(k)$ such that $(u,v)\in\Ee_F$.
By the induction hypothesis, we have $\kappa(u)\leq \frac{1}{k+1}$.
If $L(u)=\{r\}$, then $L(v)=\{r\}$ and 
$$
\kappa_F(v)=\frac{1}{2}\kappa_F(u)< \min\left\{\frac{1}{k+2},\kappa_F(u)\right\}.
$$  
The case $L(u)=(r,r^*)$ for $r\in F_m$ and $s:=r\oplus r^*\in F_{m+1}\setminus F_m$ still needs to be treated. 
Since then $\frac{1}{m+1}=\kappa_F(u)\leq \frac{1}{k+1}$, we have $m\geq k$.
Furthermore, $u$ is connected to exactly three vertices $v_1,v_2,v_3\in\Vv_F(k+1)$ such that $L(v_1):=(r,s)\in\Ll$, $L(v_2):=\{s\}\in\Ll$ and $L(v_3):=(s,r^*)\in\Ll$. 
Thus, $v\in\{v_1,v_2,v_3\}$.
By construction, we have $\kappa_F(v_j)\leq\frac{1}{m+2}<\kappa_F(u)$ for $j\in\{1,3\}$ and $\kappa_F(v_2)=\frac{1}{2}\kappa_F(u)<\min\{\frac{1}{m+2},\kappa_F(u)\}$. 
Thus, the induction step is finished by using $m\geq k$.


The previous considerations imply that $\kappa_F$ is a weight. 
Hence, $d_{\Tt_F}$ defines an ultra metric on $\partial\Tt_F$ inducing the product topology, see \cite[Proof of prop.~6]{PB09}.
Since $\partial\Tt_F$ is compact in the product topology, we conclude $(\partial\Tt_F,d_{\Tt_F})$ is a compact metric space.

\end{proof}

\begin{lemma}
\label{lem-ExMap_Tree-Farey}
Let $\gamma\in\partial\Tt_F$. Then there is a unique $x(\gamma)\in\ol{[0,1]}_F$ such that for any sequence $(\alpha_k)_{k\in\N_0}\subseteq[0,1]_F$ with $\alpha_k\in L(\gamma(k))$ for all $k\in\N_0$, we have $\lim_F\alpha_k =x(\gamma)$.
\end{lemma}

\begin{proof}
Let $\gamma\in\partial\Tt_F$ and $\alpha_k\in L(\gamma(k))$  for all $k\in\N_0$. For $k\geq 1$, we either have $L(\gamma(k))=\{r_k\}$ or $L(\gamma(k))=(r_k,r_k^*)$ with $r_k\in F_{m_k}$ and $m_k\in\N$. 
Then $\gamma(k)\in\Vv_F(k)$ implies $\kappa\big(\gamma(k)\big)\leq \frac{1}{k+1}$ by Proposition~\ref{prop-BoundaryMetric}. 
Hence, $m_k\geq k$ follows by definition of $\kappa$. 
Moreover, Lemma~\ref{lem-BasicFareyTree} states $L(\gamma(m))\subseteq L(\gamma(k))$ if $m\geq k$. Thus, $d_F(\alpha_k,\alpha_m)\leq \frac{1}{k_0+1}$ follows for all $k_0\in\N$ and $k,m\geq k_0$. Thus, $(\alpha_k)_{k\in\N_0}$ is a Cauchy sequence with respect to $d_F$ and so the limit $\lim_F \alpha_k=:x(\gamma)$ exists. 

Let $(\beta_k)_{k\in\N_0}$ be another sequence satisfying $\beta_k\in L(\gamma(k))$  for all $k\in\N_0$. Then the previous considerations imply $\alpha_k,\beta_k\in L(\gamma(k))\subseteq L(\gamma(k_0))$ and so $d_F(\alpha_k,\beta_k)\leq \frac{1}{k_0+1}$ for all $k_0\in\N$ and $k\geq k_0$. Hence, $\lim_F\beta_k=\lim_F\alpha_k=x(\gamma)$ is concluded.
\end{proof}

The previous lemma allows us to define a map $\Psi:\partial\Tt_F\to\ol{[0,1]}_F, \, \gamma\mapsto x(\gamma)$.

\begin{theo}
\label{theo-TreeIsoHomeo}
The map $\Psi: \partial\Tt_F \to \ol{[0,1]}_F, \, \gamma\mapsto x(\gamma),$ is a surjective isometry. 
In particular, the metric spaces $(\partial\Tt_F, d_{\Tt_F})$ and $(\ol{[0,1]}_F,d_F)$ are isometric isomorphic.
\end{theo}

\begin{proof}
We first prove that $\Psi$ is isometric.
Let $\gamma, \eta \in \partial\Tt_F$ be different and $v:=\gamma\land\eta=\gamma(n)\in\Vv_F$ for some $n\in\N_0$. 
Choose $(\alpha_k)_{k\in\N_0},(\beta_k)_{k\in\N_0}\subseteq[0,1]_F$ such that $\alpha_k\in L\big(\gamma(k)\big)$ and $\beta_k\in L\big(\eta(k)\big)$ for $k\in\N_0$. 

Suppose first $n=0$, i.e. $\gamma(n)=o$. Then $d_{\Tt_F}(\gamma,\eta)=\kappa_F(v)=1$ holds and $\gamma(1)\neq\eta(1)$. 
Lemma~\ref{lem-BasicFareyTree}~(b) implies $L\big(\gamma(1)\big)\cap L\big(\eta(1)\big)=\emptyset$. 
Thus, $\alpha_k\in\{0,1\}$ or $\beta_k\in\{0,1\}$ follows for $k\geq 1$. 
Hence, $d_F(\alpha_k,\beta_k)=1=d_{\Tt_F}(\gamma,\eta)$ is concluded.

Suppose now $n\geq 1$. 
Recall that $L(v)=(r,r^*)$ follows for some $r\in F_m$ and $m\in\N$ since $\gamma(n)\neq o$ and $v=\gamma \wedge \eta$ with $\gamma\neq \eta$.
By Lemma~\ref{lem-FareyNgb}~(b), there is no loss in assuming that $m\in\N$ is chosen such that $r\oplus r^*\in F_{m+1}\setminus F_m$. 
From this $d_{\Tt_F}(\gamma,\eta)=\frac{1}{m+1}$ is concluded. 
Since $L\big(\gamma(n+1)\big)\cap L\big(\eta(n+1)\big)=\emptyset$, Lemma~\ref{lem-BasicFareyTree}~(a) yields that either $\alpha_k<\beta_k$ for all $k\geq n+1$ or $\alpha_k>\beta_k$ for all $k\geq n+1$. 
Without loss of generality assume $\alpha_k<\beta_k$ for all $k\geq n+1$. 
Then the choice of $m$ and Lemma~\ref{lem-FareyNgb}~(b) imply $r < \alpha_k\leq r\oplus r^* \leq \beta_k < r^*$ for $k\geq n+1$. 
Hence, $d_F(\alpha_k,\beta_k)=\frac{1}{m+1}$ follows for $k\geq n+1$. 
Since $\lim_F\alpha_k=x(\gamma)$ and $\lim_F\beta_k=x(\eta)$ holds by Lemma~\ref{lem-ExMap_Tree-Farey}, we conclude
$$
d_F(x(\gamma),x(\eta)) = \lim_{k\to\infty}d_F(\alpha_k,\beta_k)=\frac{1}{m+1} = d_{\Tt_F}(\gamma,\eta).
$$
Next, we show that $\Psi$ is surjective. 
Note that by the previous considerations $\Psi$ is continuous and $(\partial\Tt_F,d_{\Tt_F})$ is compact by Proposition~\ref{prop-BoundaryMetric}. 
Thus, the image $\Psi\big(\partial\Tt_F\big) \subseteq \ol{[0,1]}_F$ is compact in $(\ol{[0,1]}_F,d_F)$ and so in particular it is closed. 
Therefore, it suffices to show that the image is dense in $\ol{[0,1]}_F$. 
Let $y\in \ol{[0,1]}_F$ and $\varepsilon>0$. 
Then there is an $\alpha\in[0,1]_F$ such that $d_F(y,\alpha)<\varepsilon$. 
Let $k_0\in\N$ be such that $\frac{1}{k_0+1}<\varepsilon$ and choose $\gamma\in\partial\Tt_F$ such that $\alpha\in L(\gamma(k_0))$, which exists by Lemma~\ref{lem-BasicFareyTree}~(b). 
Let $\beta_k\in L(\gamma(k))$ for all $k\in\N_0$. 
Then $\lim_F\beta_k=x(\gamma)$. 
Furthermore, for $k\geq k_0$, we have $\alpha,\beta_k\in L(\gamma(k_0))$ by Lemma~\ref{lem-BasicFareyTree}~(a). 
Since $\kappa\big(\gamma(k_0)\big)\leq \frac{1}{k_0+1}$ by Proposition~\ref{prop-BoundaryMetric}, we conclude $d_F(\alpha,\beta_k)\leq \frac{1}{k_0+1}<\varepsilon$ for all $k\geq k_0$. 
Hence, $d_F(\alpha,x(\gamma))=\lim_{k\to\infty}d_F(\alpha,\beta_k)<\varepsilon$. 
Since $d_F$ is an ultra metric (Proposition~\ref{prop-dF_metric}), the estimate $d_F(y,x(\gamma))<\varepsilon$ follows. 
This proves that the image $\Psi\big(\partial\Tt_F\big) \subseteq \ol{[0,1]}_F$ is dense and closed and so it coincides with $\ol{[0,1]}_F$.
\end{proof}

%
%


\section{Defects}
\label{sect-defects}

In \cite{BIT91}, spectral defects were discovered in the Kohmoto butterfly at all rational points described via localized impurities of the corresponding media. 
In this work, we can identify these localized impurities explicitly on the level of the configurations.
They are described in terms of the corresponding Farey neighbors through the previously developed theory.

Thanks to Theorem~\ref{theo-Homeo_Farey-Subshifts}, we have a unique isometric surjetive map $\Phi:\ol{[0,1]}_F\to\ol{\QS}$ such that $\Phi(\alpha)=\Omega_\alpha=\ol{\Orb(\omega_\alpha)}$ for all $\alpha\in[0,1]$, where $\omega_\alpha\in\Aa^\ZM$ is the mechanical word associated with $\alpha$. 
Recall that $\Aa=\{0,1\}$. 
In this section, we will show that the subshift $\Phi(x)$ for $x\in\ol{[0,1]}_F\setminus [0,1]$ is also topological transitive, namely that there is a $\omega_x\in\Aa^\ZM$ such that $\Phi(x)=\ol{\Orb(\omega_x)}$. 

Due to Proposition~\ref{prop-convergence_in_farey_topo}, there is a unique rational number $r=\frac{p}{q}\in[0,1]$ (with $p$ and $q$ coprime) for each $x\in\ol{[0,1]}_F\setminus [0,1]$ such that $x=r_-$ or $x=r_+$.
We prove in this section that the configuration $\omega_x$ is then uniquely determined by $r$ and one of its $q$-Farey neighbors $r_*$ or $r^*$.

These Farey neighbors can be described using the concept of continued fraction expansion. 
A short reminder is presented here and more details can be found in \cite{S81}. 
Every rational number $r\in \QM$ can be represented by a finite string of integers, namely
$$
r = a_0 + \frac{1}{a_1 + \frac{1}{a_2 + \ldots \underset{+\frac{1}{a_n}}{\,}}} =: [0,a_0,a_1,a_2,\ldots,a_n]
$$
with $a_0 \in \ZM , a_i\in\N\setminus\{0\}$.
Note that the additional ``$0$'' at the beginning of the string is artificial following the convention used in \cite{BBL24,BaBeBiRaTh24}, which will be usefull in Sections~\ref{sect-SpectralApplic} and \ref{sect-SpectralMapContOpt}.
This representation is not unique since the strings $[0,a_0, \dots a_{n}, 1]$ and $[0,a_0, \dots a_{n}+1]$ define the same rational number. 



Specifically, given a rational point $r \in (0,1)\cap\QM $, there are precisely two associated tuples
$$
\ca(r) := [0,a_0 , \dots, a_n + 1] \quad \text{and} \quad \cb(r) := [0,a_0, \dots, a_n, 1],
$$   
such that their \textit{evaluation} equals $r$, that is
$$
r = a_0 + \frac{1}{a_1 + \frac{1}{a_2 + \ldots \underset{+\frac{1}{a_n+1}}{\,}}}  = a_0 + \frac{1}{a_1 + \frac{1}{a_2 + \ldots \underset{+\frac{1}{a_n +\frac{1}{1}}}{\,}}} . 
$$
As a consequence, we always have $a_0 = 0$ for both $\ca(r)$ and $\cb(r)$.
We call $\ca(r)$ the associated {\em short-} and $\cb(r)$ the {\em long continued fraction} of the rational $r\in(0,1)\cap\QM$.

The points $r=0$ or $r=1$ are different in the sense that they only have one $1$-Farey neighbor, i.e. $1_*=0$ and $0^*=1$. 
Therefore these points usually need a special treatment.
We associate the string $\ca = [0,0]$ with $r=0$ and the string $\ca = [0,0,1]$ with $r=1$.

The Farey neighbors are precisely described through the continued fraction.

\begin{lemma}[{\cite[lem.~2]{S81}}]
\label{lem-FarNghCF}
Let $r=\frac{p}{q}\in (0,1)\cap \QM$ be such that $p$ and $q$ are coprime. If $\cb(r)=[0,a_0, a_1,\ldots,a_n,1]$, then the $q$-Farey-neighbors $r_*,r^*\in F_q$ of $r$ satisfy the following.
\begin{itemize}
\item[(a)] If $n\in\N$ is odd, then $r_*$ is the evaluation of $[0,a_0,a_1,\ldots, a_{n-1}]$ and $r^*$ is the evaluation of  $[0,a_0, a_1,\ldots, a_n]$.
\item[(b)] If $n\in\N$ is even, then $r_*$  is the evaluation $[0,a_0,a_1,\ldots, a_n]$ and $r^*$ is the evaluation of $[0,a_0, a_1,\ldots, a_{n-1}]$.
\end{itemize}
\end{lemma}

\begin{exam}
The short and long continued fraction of $r = \frac{1}{2}$ are given by $\ca(r) = [0,0,2]$ and $\cb(r) = [0,0,1,1]$. 
Thus, its Farey neighbors in $F_2$ are given by $r_* = 0$ (the evaluation of $[0,0]$) and $r^* = 1$ (the evaluation of $[0,0,1]$). 
For $r = \frac{2}{3}$, we have $\ca(r) = [0,0,1,2]$ and $\cb(r) = [0,0,1,1,1]$. 
Hence, its Farey neighbors in $F_3$ are given by $r_* = \frac{1}{2}$ (the evaluation of $[0,0,1,1]$) and $r^* = 1$ (the evaluation of $[0,0,1]$).
\end{exam}

In the following, we use the notation 
$$
\omega|_{ \{n \in \ZM: n < 0\} } \cdot \omega|_{ \{n \in \ZM: n \geq 0\} } = \omega \in \csp
$$ 
where the dot indicates the origin.
For $s=s(1)\ldots s(n) \in \Aa^n$ and $m \in \N$, we denote its $m$-fold repetition as $s^m \in \Aa^{nm}$ and define its {\em two-sided periodic repetition} $s^\infty$ by 
$$ 
s^\infty := \dots sss\cdot sss\dots.
$$
In particular, we have $s^\infty(0) = s(1)$. 
For $v \in \Aa^\ast := \bigcup_{n \in \N_0 } \Aa^n$, we  define $s^\infty v \cdot s^\infty$ by
$$
s^\infty v \cdot s^\infty := \dots sss v \cdot sss \dots ~. 
$$
We call $s^\infty v \cdot s^\infty$ a {\em finite defect of $s^\infty$ with impurity $v$}.

\begin{lemma}
\label{lem:HausdorffConv_Impurity}
For $u,v\in\Aa^\ast\setminus \{ \o\}$, let $\omega_m:=(u^m v)^\infty$. Then the limit $\lim_{m\to\infty}\Orb(\omega_m)$ exists in $(\Inv,d_\Inv)$ and equals to $\overline{\Orb(u^\infty v\cdot u^\infty)}$.
\end{lemma}
\begin{proof}
Observe $\Orb(\omega_m)=\overline{\Orb(\omega_m)}\in\Inv$. Then the equality $\lim_{m\to\infty}\Orb(\omega_m)=\overline{\Orb(u^\infty v\cdot u^\infty)}$ follows by a straightforward computation using Proposition~\ref{prop-TopHau} and $\lim\limits_{m\to\infty}|u^m|=\infty$, where $|u^m|$ denotes the length of the word $u^m$.
\end{proof}

Given the finite continued fraction $ \c =  [0,a_0, a_1, \dots , a_n]$, we recursively define $s_k = s_k[\c] \in \Aa^\ast$  by 
$$
s_{-1} := 1, \quad  s_0 := 0, \quad s_1 := s_0^{a_1-1}s_{-1}, \quad s_k := s_{k-1}^{a_k}s_{k-2} ~ \text{ for }  2 \leq k \leq n.
$$
We omit the $c$ dependence of $s_k$ whenever it is clear from context.
Observe that the word $s_k[\c]$ depends only on the first $k$ digits of the continued fraction $\c$.
The following well-known facts can for instance be found in \cite[lem.~2.4]{BaBeBiRaTh24} and references therein.

\begin{lemma} \label{lem-RecWord}
Let $r=\frac{p}{q} \in (0,1) \cap \QM$ with $p$ and $q$ coprime and $n\in\N_0$ be such that $\ca(r)$ is an $n$-tuple.
Then the following assertions hold.
\begin{enumerate}[(a)]
\item We have $s_n[\ca(r)],s_{n+1}[\cb(r)] \in\Aa^q$.
\item The configuration $s_n[\ca(r)]^\infty$ is a shift of $s_{n+1}[\cb(r)]^\infty$. 
\item We have $\Omega_r = \Orb(s_n[\ca(r)]^\infty)= \Orb(s_{n+1}[\cb(r)]^\infty)$.
\end{enumerate}
\end{lemma}

\begin{proof}
Statement (a) and (c) are well-known, see e.g. \cite[sec.~2.2]{BaBeBiRaTh24}. 
For (b), observe that $s_k:=s_{k}[\ca(r)]=s_{k}[\cb(r)]$ for $k\leq n-1$. 
Thus, $s_n[\ca(r)] = s_{n-1}^{a_n+1} s_{n-2}$ and 
$s_{n+1}[\cb(r)] = s_{n-1}^{a_n} s_{n-2}s_{n-1}$ hold using the recursive definition. Hence, 
\begin{align*}
s_n[\ca(r)]^\infty 
	= \big( s_{n-1}^{a_n+1} s_{n-2}\big)^\infty \cdot \big( s_{n-1}^{a_n+1} s_{n-2}\big)^\infty
	= &s_{n+1}[\cb(r)]^\infty s_{n-1}^{a_n} s_{n-2} \cdot s_{n-1} s_{n+1}[\cb(r)]^\infty \\
	= &\tra^{-m}\big( s_{n+1}[\cb(r)] ^\infty \big)
\end{align*}

follows where $m\in\N_0$ is the length of the word $s_{n-1}$.
\end{proof}


With this at hand, we can describe the image $\Phi(r_\pm)$ for $r\in[0,1]\cap\QM$ and $r_\pm=\lim_F r\pm \frac{1}{n}\in\ol{[0,1]}_F$ where $\Phi:\ol{[0,1]}_F\to\ol{\QS}$ is the surjective isometry defined in Theorem~\ref{theo-Homeo_Farey-Subshifts}.
We first treat the special cases $r=0$ and $r=1$ as a warmup in the following example.

\begin{exam}
\label{exam:Defects_0_&_1}
For technical reasons we excluded the case $r=0$ with continued fraction expansion $\ca(r)=[0,0]$ and $s=1$ with continued fraction expansion $\ca(s)=[0,0,1]$. 
We discuss them here.
First note that $\omega_r=0^\infty$ and $\omega_s=1^\infty$.
Furthermore, $r,s$ are $1$-Farey neighbors, namely $r^\ast=s$ and $s_\ast=r$.
Define 
\[
\omega_{0_+} := 0^\infty 1 \cdot 0^\infty
\qquad\textrm{and}\qquad
\omega_{1_-} := 1^\infty 0 \cdot 1^\infty\,.
\]
Let $r_m$ be the evaluation of $\c_m:=[0,0,m]$ and $s_m$ be the evaluation of $\d_m=[0,0,1,m]$. 
Then $0<r_m$ and $s_m<1$ for $m\in\N$, $\lim_F r_m=0_+$ and $\lim_F s_m=1_-$.
In addition, $s_1[\c_m]=0^{m-1}1$ and $s_2[\d_m]=1^m 0$ hold for $m\in\N$.
With this at hand, Lemma~\ref{lem:HausdorffConv_Impurity} leads to 
$$
\Phi(0_+)=\ol{\Orb(\omega_{0_+})}
\qquad\textrm{and}\qquad
\Phi(1_+)=\ol{\Orb(\omega_{1_+})}.
$$
\end{exam}

\begin{proposi}
\label{prop-Defects}
Let $r = \frac{p}{q} \in (0,1)\cap \QM$ with $p$ and $q$ are coprime and consider the $q$-Farey neighbors $r_*<r<r^*$. Then there are finite defects $\omega_{r_-}, \omega_{r_+} \in \csp$ of $\omega_r$ whose impurities depend only on $r_* \in F_q$  and $r^* \in F_q$ such that $\Omega_{r_\pm} := \Phi(r_\pm)= \ol{\Orb(\omega_{r_{\pm}})}$. 
In particular, the following assertions hold.
\begin{enumerate}[(a)]
\item The subshift $\Omega_{r_\pm}$ is topological transitive but not minimal. In addition, $\Omega_r\subsetneq \Omega_{r_\pm}$.
\item For $\ca(r)=[0,a_0,\ldots,a_n+1]$ and $\cb(r)=[0,a_0,\ldots,a_n,1]$, define
\begin{align*}
\omega_s &:= (s_n[0,a_0, \dots, a_n +1])^\infty  s_{n-1}[0,a_0, \dots, a_{n-1}] \cdot (s_n[0,a_0, \dots, a_n +1])^\infty,\\
\omega_l &:= (s_{n+1}[0,a_0, \dots, a_n ,1])^\infty  s_{n}[0,a_0, \dots, a_{n}] \cdot (s_{n+1}[0,a_0, \dots, a_n, 1])^\infty.
\end{align*}
Then we have $\omega_{r_+} = \begin{cases} \omega_s, & n \text{ even,} \\ \omega_l,  & n \text{ odd,}  \end{cases}$ and $\omega_{r_-} = \begin{cases} \omega_l, & n \text{ even,} \\ \omega_s,  & n \text{ odd.}  \end{cases}$
\end{enumerate}
\end{proposi}

\begin{rem}
Let $r \in (0,1)\cap \QM$ with $\ca(r)=[0,a_0,\ldots,a_n+1]$ and $\cb(r)=[0,a_0,\ldots,a_n,1]$. Then 
\begin{itemize}
\item $\omega_s$ is a finite defect of $\omega_r$ with impurity $s_{n-1}[0,a_0, \dots, a_{n-1}]$, and
\item $\omega_l$ is a finite defect of $\omega_r$ with impurity $s_{n}[0,a_0, \dots, a_{n}]$.
\end{itemize}
\end{rem}

\begin{proof}
Let $\ca(r) = [0,a_0, \dots, a_n + 1]$ and $\cb(r) = [0,a_0, \dots, a_n , 1]$ be the short and long continued fractions associated with $r$.
We only treat the case if $n$ is even. 
The odd case follows similarly. Note that the switch in statement (b) between the even and odd case comes from the switch in Lemma~\ref{lem-FarNghCF}. 

For $r_+$ and $m\in\N$, let $r_m$ be the evaluation of $\c_m:=[0,a_0, \dots, a_n + 1, m + 1]$. 
Observe $\lim_E r_m = r$ and $r< r_m$ holds by Lemma~\ref{lem-FarNghCF} using that $n$ is even.
Thus, $\lim_F r_m = r_+$ follows from Proposition~\ref{prop-convergence_in_farey_topo}. 
Then Theorem~\ref{theo-Homeo_Farey-Subshifts} implies $\lim_{m\to\infty} \Omega_{r_m} = \Omega_{r_+}$.
Lemma~\ref{lem-RecWord} asserts $\Omega_{r_m} = \ol{\Orb(s_{n+1}[\c_m]^\infty)}$ where $s_{n+1}[\c_m] = s_{n}[\c_m]^{m+1}s_{n-1}[\c_m]$. 
By the recursive definition, $s_{n}[\c_m]$ depends only on the first $n$ digits of $\c_m$, which happens to be the continued fraction $\ca(r)$.
Thus, $s_{n}[\c_m] = s_{n}[\ca(r)]$.
Since $n$ was even, Lemma~\ref{lem-FarNghCF} implies that $r^*$ is the evaluation of $\c^* := [0,a_0, \dots, a_{n-1}]$ and by a similar argument as before we have $s_{n-1}[\c_m] = s_{n-1}[\c^*]$.
Therefore we conclude
$$
s_{n+1}[\c_m] =s_{n}[\c_m]^{m+1}s_{n-1}[\c_m]= s_{n}[\ca(r)]^{m+1}s_{n-1}[\c^*].
$$
Taking periodic repetitions of this word and sending $m$ to infinity then yields
$$
\lim_{m\to \infty} s_{n+1}[\c_m]^\infty = s_{n}[\ca(r)]^{\infty} s_{n-1}[\c^*] \cdot s_{n}[\ca(r)]^{\infty} =: \omega_{r_+}
$$
Thus, Lemma~\ref{lem:HausdorffConv_Impurity} and Lemma~\ref{lem-RecWord}~(c) imply
$$
\ol{\Orb(\omega_{r_+})}
	= \lim_{m \to \infty} \ol{\Orb(s_{n+1}[\c_m]^\infty)} 
	= \lim_{m \to \infty} \Omega_{r_m}
	= \Omega_{r_+}.
$$
Note that $s_{n}[\ca(r)]\in\Aa^q$ by Lemma~\ref{lem-RecWord} and so $\lim_{n\to\infty}\tra^{qn}(\omega_{r_+})=\omega_r$. 
Hence, $\Omega_r\subseteq\Omega_{r_+}$ follows. 

\medskip

For $r_-$ and $m\in\N$, we proceed analogously. Let $r_m$ be the evaluation of the tuple $\c_m:=[0,a_0, \dots, a_n, 1, m+1]$.
Then $\lim_F r_m = r_-$ holds since $n$ is even. 
Moreover, $n$ even implies that $r_*$ is the evaluation of $\c_* = [a_0, \dots, a_n]$ by Lemma~\ref{lem-FarNghCF}.
Thus, $\c_m$ and $\c_*$ coincide on the first $n$ digits yielding $s_n[\c_m] = s_n[\c_*]$.
Similarly, we have $s_{n+1}[\c_m]=s_{n+1}[\cb(r)]$.
Hence, we get
\begin{align*}
s_{n+2}[\c_m] 
	&= s_{n+1}[\c_m]^{m+1} s_{n}[\c_m] 
	= s_{n+1}[\cb(r)]^{m+1} s_n[\c_*].
\end{align*}
Taking periodic repetitions of this word and sending $m$ to infinity leads to 
$$
\lim_{m\to \infty}s_{n+2}[\c_m]^\infty = s_{n+1}[\cb(r)]^\infty s_n[\c_*] \cdot s_{n+1}[\cb(r)]^\infty =: \omega_{r_-}.
$$
As before, Lemma~\ref{lem:HausdorffConv_Impurity} and Lemma~\ref{lem-RecWord}~(c) imply
$$
\ol{\Orb(\omega_{r_-})}
	= \lim_{m \to \infty} \ol{\Orb(s_{n+2}[\c_m]^\infty)} 
	= \lim_{m \to \infty} \Omega_{r_m}
	= \Omega_{r_-}
$$ 
and $\Omega_r\subseteq\Omega_{r_-}$ follows.
\end{proof}


For convenience of the reader, let us provide some further examples in the following table. 

\medskip

{
\begin{center}
\def\arraystretch{1.2}
\begin{tabular}{ |c | c | c | c| }
\hline
 $r$ & $\frac{1}{2}$ & $\frac{2}{3}$ & $\frac{7}{9}$\\ 
 $\ca(r)$ & $[0,0,2]$ & $[0,0,1,2]$ & $[0,0,1,3,2]$ \\  
 $n$ & $1$ & $2$ & $3$ \\  \hline
 $s_{n-1}(\ca(r))$  & $0$ & $1$ & $1110$\\
 $s_n(\ca(r))$ & $01$ & $110$ & $111011101$\\
 $s_{n}(\cb(r))$ & $1$ & $10$ & $11101$\\
 $s_{n+1}(\cb(r))$ & $10$ & $101$ & $111011110$\\
 \hline
 $\omega_s$ & $(01)^\infty 0 \cdot (01)^\infty$ & $(101)^\infty 10 \cdot (101)^\infty$ &$(111011101)^\infty 1110 \cdot (111011101)^\infty$ \\
 $\omega_l$ & $(10)^\infty 1 \cdot (10)^\infty$ & $(110)^\infty 1 \cdot (110)^\infty$  &$(111011110)^\infty 11101 \cdot (111011110)^\infty$ \\
 \hline
\end{tabular}
\end{center}
}

\medskip

\begin{defini}
\label{def:omega_x}
For $x\in\ol{[0,1]}_F$, define $\omega_x\in\Aa^\ZM$ by
$$
\omega_x(n) :=
	\begin{cases}
		\chi_{[1-\alpha,1)}(n\alpha\mod\, 1), \qquad &x\in[0,1],\\
		\omega_{r_+}(n), \qquad &x=r_+ \textrm{ for some } r\in[0,1]\cap\QM,\\
		\omega_{r_-}(n), \qquad &x=r_- \textrm{ for some } r\in[0,1]\cap\QM,
	\end{cases}
	\qquad n\in\N.
$$
\end{defini}

\begin{coro}
\label{cor:BijectionFarey_Subshifts}
The map
$$
\Phi:\ol{[0,1]}_F\to \ol{\QS},\quad 
	x\mapsto \ol{\Orb(\omega_x)}
$$
is an isometric surjective map.
\end{coro}

\begin{proof}
This follows immediately by Corollary~\ref{cor:Completion_FareySpace}, Theorem~\ref{theo-Homeo_Farey-Subshifts} and Proposition~\ref{prop-Defects}.
\end{proof}


\section{Spectral applications}
\label{sect-SpectralApplic}

\noindent We prove Theorem~\ref{theo-SpeCon} on the Lipschitz continuity of the spectral map for the class of strongly pattern equivariant Hamiltonian. As a consequence of the classification of the defects, we provide a different proof of the result obtained in \cite{BIT91}.

Consider a family of continuous functions $t_k:\Aa^\ZM\to\CM$ for $k\in K$ where $K\subseteq \ZM$ is finite. With such a (finite) family of continuous functions, we associate the operator family $H=(H_\omega)_{\omega\in\Aa^\ZM}$ of linear bounded operators $H_\omega:\ell^2(\ZM)\to\ell^2(\ZM)$ defined by
\begin{equation}
\label{eq:Hamil}
(H_\omega \psi\big)(n) :=
	\sum_{k\in K} t_k\big(\tra^{-n}(\omega)\big) \ \psi(n+k) 
	\,,\qquad \psi\in\ell^2(\ZM)\,,\; n\in\ZM\,.
\end{equation}
We call $H = (H_\omega)_{\omega \in \Omega}$ a \textit{finite range Hamiltonian} since $K$ is finite.
A function $t:\Aa^\ZM\to\CM$ is called {\em strongly pattern equivariant} if there exists an $r\in\N$ such that $\omega|_{[-r,r]}=\rho|_{[-r,r]}$ implies $t(\omega)=t(\rho)$. 
Note that $t:\Aa^\ZM\to\CM$ is strongly pattern equivariant if and only if it is continuous and it has finite range. 
Then $H = (H_\omega)_{\omega \in \Omega}$ is called a {\em strongly pattern equivariant Hamiltonian} if $H_\omega$ is self-adjoint for each $\omega\in\Aa^\ZM$ and $t_k$ is strongly pattern equivariant for all $k\in K$. 
Note that $H_\omega$ is self-adjoint if $K=-K$ and $t_{-k}(\omega)=\overline{t_{k}\big(\tra^k(\omega)\big)}$ for all $k\in K$ and $\omega\in\Aa^\ZM$.

An example of a strongly pattern equivariant Hamiltonian is the Kohmoto model -- the main object of interest in this work. Specifically, for fixed $V\neq 0$ and  $\Aa=\{0,1\}$, let $K=\{-1,0,1\}$ and define 
\[
t_{\pm 1}:\Aa^\ZM\to\RM,\quad t_{\pm 1}(\omega)=1,
\qquad\textrm{and}\qquad
t_0:\Aa^\ZM\to\RM,\quad t_0(\omega):= V \omega(0).
\]
Clearly, these functions are strongly pattern equivariant (taking only the values $0,1$ and $V$). Observe that $H_\alpha$ defined in Equation~\eqref{eq:KohmOperat} equals to $H_{\omega_\alpha}$ in this case. Here $\omega_\alpha\in\Aa^\ZM$ is the configuration $\omega_\alpha(n)=\chi_{[1-\alpha,1)}(n\alpha\mod\, 1)$ for $n\in\ZM$.

We call an operator family $\Hh = (\Hh_\omega)_{\omega \in \Omega}$ a \textit{Hamiltonian of infinite range} if it is self-adjoint and for each $\varepsilon > 0$, there is a finite range Hamiltonian $H = (H_\omega)_{\omega \in \Omega}$ such that
$$ 
\|\Hh - H\| := \sup_{\omega \in \Omega} \|\Hh_\omega - H_\omega\| \leq \varepsilon.
$$
This class of operators resembles elements of the Roe algebras \cite{Ro88,Ro93}. 
Such operators are in particularly given by operators of the form \eqref{eq:Hamil} where $K$ is allowed to be infinite but where the coefficients $t_k$ decay sufficiently fast if $|k|\to\infty$.
Note that every Hamiltonian of finite range is also a Hamiltonian of infinite range.

For $m\in\ZM$, consider the unitary operator $U_m\in\Ll(\ell^2(\ZM)), \big(U_m\psi\big)(n)=\psi(n-m)$, the shift operator. 
Let $H=(H_\omega)_{\omega\in\csp}$ be a Hamiltonian of finite or infinite range. 
Then this operator family is $\ZM$-covariant, i.e. $U_m H_\omega U_{-m}=H_{\tra^m(\omega)}$ for $m\in\ZM$ and $\omega\in\Aa^\ZM$. 
If $H$ is of finite range, then $\csp\ni \omega\mapsto H_\omega$ is strongly continuous as the coefficients $t_k$ of $H$ are continuous. 
If $\Hh$ is of infinite range, then strong continuity follows also by standard arguments and that the operators are approximated in operator norm by Hamiltonians of finite range.

For $\Omega\in\Inv$, the operator family $\Hh_\Omega:= (\Hh_\omega)_{\omega\in\Omega}$ is a {\em Hamiltonian associated with $\Omega$} and the spectrum of $\Hh_\Omega$ is defined by $\sigma(\Hh_\Omega):=\overline{\bigcup_{\omega\in\Omega}\sigma(\Hh_\omega)}$.
Since the operators $\Hh_\omega$ are self-adjoint, the strong continuity and the covariance imply $\sigma(\Hh_\Omega)=\sigma(\Hh_\omega)$ if $\Omega = \ol{\Orb(\omega)}$ for $\omega\in\csp$. 

For $x\in\ol{[0,1]}_F$, let $\omega_x\in\Aa^\ZM$ be defined as in Definition~\ref{def:omega_x}.
As a consequence, we obtain the following.

\begin{theo}
\label{theo-SpeCon_PatternEquiv}
Let $H$ be a strongly pattern equivariant Hamiltonian. Then the spectral map
$$\Sigma:\overline{[0,1]}_F\to\Kk(\RM)
	\,,\quad x\mapsto\sigma\big(H_{\omega_x}\big)\,,
$$
is Lipschitz-continuous. 
In particular there is some $C > 0$ such that 
$$
d_H(\sigma(H_{\omega_x}), \sigma(H_{\omega_y}))  \leq C \cdot d_F(x,y),
	\qquad x,y \in \ol{[0,1]}_F.
$$
\end{theo}

\begin{proof}
For $\Omega_x:=\ol{\Orb(\omega_x)}$, we have $\sigma(H_{\Omega_x})=\sigma(H_{\omega_x})$.
Then \cite[prop.~1.1]{BBC18} respectively \cite[thm.~1.3~(b)]{BT21}, imply that there exists a $C > 0 $ (depending also on the coefficients of $H$) such that
\[
d_H(\sigma(H_{\omega_x}), \sigma(H_{\omega_y}))
	= d_H(\sigma(H_{\Omega_x}), \sigma(H_{\Omega_y}))
	\leq C\cdot d_\Inv(\Omega_x, \Omega_y).
\]
Thus, Corollary~\ref{cor:BijectionFarey_Subshifts} immediately yields the claimed statement. 
\end{proof}

Note that the restriction to strongly pattern equivariant Hamiltonians is just for simplicity while it covers the Kohmoto model. Similar results hold as long as the operators are normal, the off-diagonals decay fast enough and the coefficients are regular enough, see \cite{BT21} for details.

\begin{coro}
\label{cor-LipCon_PatternEquiv}
Let $H$ be a strongly pattern equivariant Hamiltonian (of finite range). 
{Then there is a constant $C>0$ such that following holds.}
\begin{itemize} 
\item[(a)] Let $r\in F_m$ and $x,y\in\overline{(r,r^*)}_F$, then
$
d_H\big(\sigma(H_{\omega_x}),\sigma(H_{\omega_y})\big) 
	\leq \frac{C}{m+1}
$.
\item[(b)] Let $\alpha,\frac{p}{q}\in[0,1]$ be such that $0<| \alpha-\frac{p}{q} | < \frac{1}{q^2}$, then 
$
d_H\big(\sigma(H_\alpha),\sigma(H_{\frac{p}{q}})\big) 
	\leq  \frac{C}{q}
$.
\end{itemize}
\end{coro}

\begin{proof}
(a)  This is a direct consequence of Theorem~\ref{theo-SpeCon_PatternEquiv} and the fact that $d_F(x,y)\leq \frac{1}{m+1}$ if $x,y\in\overline{(r,r^*)}_F$ for $r\in F_m$, see Section~\ref{ssect-FareyTopoRep}.

\medskip

(b)  Clearly, $r:=\frac{p}{q}\in(0,1)$ is a $q$-Farey number. Since $0<|\alpha-r|<\frac{1}{q^2}$, $\alpha$ is not a $q$-Farey number by Lemma~\ref{lem-FareyDist}. Hence, either $\alpha\in(r_*,r)$ or $\alpha\in(r,r^*)$ follows. Since $r\not\in F_{q-1}$, we conclude $d_F(r,\alpha)=\frac{1}{q}$ from the definition of the Farey metric $d_F$, see Section~\ref{ssect-FareyTopoRep}. Then Theorem~\ref{theo-SpeCon_PatternEquiv} leads to the desired estimate.
\end{proof}

%
%

Standard arguments imply that the essential spectrum $\sigma_{\mathrm{ess}}(\Hh_\omega)$ of an infinite range Hamiltonian $\Hh$ does not change if $\omega$ is replaced by a finite defect of it. Combined with Proposition~\ref{prop-Defects}, we obtain the analog of \cite[thm.~4, thm.~6]{BIT91}. 
%

\begin{proposi}
\label{prop:EssentialSpectrum}
Let $\Hh$ be a Hamiltonian of infinite range.
For every $u,v \in \Aa^\ast \setminus \{\o \}$ we get
$$
\sigma(\Hh_{u^\infty}) 
	= \sigma_{\mathrm{ess}}(\Hh_{u^\infty}) 
	= \sigma_{\mathrm{ess}}(\Hh_{u^\infty v\cdot u^\infty}). 
$$
In particular, for every $r \in [0,1]\cap\QM$, we have $\sigma(\Hh_{\omega_r})=\sigma_{\mathrm{ess}}(\Hh_{\omega_r}) = \sigma_{\mathrm{ess}}(\Hh_{\omega_{r_\pm}}).$
\end{proposi}

\begin{proof}
This result follows by standard arguments using that the impurity $v$ in $u^\infty$ defines a compact perturbation on the level of the operators. 
Then the essential spectrum is characterized by Weyl sequences. 
For convenience, of the reader, we provide some more details.

Note that our infinite range Hamiltonians form a subclass of the \textit{band-dominated operator} in \cite{BLLS18}.
Therefore \cite[cor.~12]{LiSe14} (see also \cite[prop.~1.2]{BLLS18}) applied to the operator $\Hh_{u^\infty v\cdot u^\infty}$ yields
\[
\sigma_{\mathrm{ess}}(\Hh_{u^\infty v\cdot u^\infty}) 
	= \bigcup_{m\in \N} \sigma_{\mathrm{ess}}(\tra^m\Hh_{u^\infty}) 
	= \sigma_{\mathrm{ess}}(\Hh_{u^\infty})
	= \sigma(\Hh_{u^\infty}) ,
\]
where the last equality follows from standard minimality arguments of $\ol{\Orb(u^\infty)}$.

Finally, Proposition~\ref{prop-Defects} asserts that the configuration $\omega_{r_\pm}$ is a defect with some impurity $v_\pm\in \Aa^\ast$ of $\omega_r$. Thus, the previous considerations yield the claimed identity.
\end{proof}

We finish this section by making the previous statement more precise in the Kohmoto model and explaining the relative positions of the spectral defects and spectral bands seen in Figure~\ref{fig-Kohmoto}. 
For convenience, we use the notation
$$
\sigma_x:=\sigma_{x}(V):=\sigma(H_{\omega_x,V}), \qquad x\in\ol{[0,1]}_F, V\neq 0,
$$
for the spectrum of the self-adjoint operator $H_{\omega_x,V}$ defined in Equation~\eqref{eq:SchroedOper}. 
We show that the difference $\sigma_{r_\pm}(V)\setminus\sigma_{r}(V)$  contains exactly $q$ points if $V\neq 0$ and $r=\frac{p}{q}$ with $p$ and $q$ coprime.
Moreover, we can determine in which gaps these eigenvalues appear proving the numerical observations discussed in the last remark in \cite{BIT91}. 
If $V>4$, this follows quickly from our considerations and \cite{Raym95}, see also \cite{BaBeBiRaTh24} a recent review of this work. 
In order to show the result in the small-coupling regime, we need the new ideas developed in \cite{BBL24}. 
Specifically, the following lemma is motivated by \cite[lem.~7.16]{BBL24}. 

We try to suspress most of the terminology used there and the reader is referred to \cite{BBL24,BaBeBiRaTh24} for a detailed discussion. 
A connected component $I=[a,b]$ of $\sigma_r$ with $r\in[0,1]\cap\QM$ and $a<b$ is called a {\em spectral band} and $a,b$ are called the associated spectral band edges. 
We say an interval $[a,b]$ {\em is to the left of an interval} $[c,d]$ if $a<c$ and $b<d$.
In this case, we write $[a,b]\prec[c,d]$ and we include the case of intervals being singletons $[a,a]=\{a\}$.

\begin{lemma}
\label{lem:Spectrum_event_disjoint}
Let $V\neq 0$, $r=\frac{p}{q}\in[0,1]\cap\QM$ with $p$ and $q$ coprime and $c_r:=[0,a_0, a_1, \dots , a_{n}]$ be its (long or short) continued fraction. 
For $k\in\N$, let $r_k\in[0,1]$ be the unique rational number defined by the continued fraction $d_k:=[0,a_0, a_1, \dots , a_{n},k]$. Then the following assertions hold.
\begin{enumerate}[(a)]
\item The spectrum $\sigma_r(V)$ equals to the union of $q$ disjoint spectral bands $J_1\prec J_2 \prec \ldots\prec J_q$.
\item For each $k\in\N$, there are exactly $q$-spectral bands $I_1^k,I_2^k,\ldots,I_q^k$ of $\sigma_{r_k}(V)$ such that 
$$
I_j^{k+1}\subseteq I_j^k
\qquad\textrm{and}\qquad
I_j^k\not\subseteq \sigma_r(V),
\qquad 1\leq j\leq q,\; k\in\N.
$$
\item For each $k\in\N$, we have
$$
\qquad\quad n\in\N\textrm{ odd} 
	\qquad\Longrightarrow\qquad
I_1^k \prec J_1 \prec I_2^k \prec J_2 \prec \ldots \prec I_q^k \prec J_q
$$
and
$$
n\in\N\textrm{ even or } n=0 
	\qquad\Longrightarrow\qquad
J_1 \prec I_1^k \prec J_2 \prec \ldots \prec I_{q-1}^k \prec J_q \prec I_q^k.
$$
\item There exists a $k_0\in\N$ such that for all $k\geq k_0$, and each $1\leq j\leq q$, we have $I_j^k\cap \sigma_r(V)=\emptyset$.
\end{enumerate}
\end{lemma}

\begin{proof}
This proof is heavily based on the structure of the spectrum proven in \cite{BBL24} (respectively \cite{Raym95} for $V>4$). 
We borrow some terminologies introduced in the work, see also \cite{BaBeBiRaTh24}.
As discussed in \cite[thm.~1.9~(e)]{BBL24}, we can restrict to the case $V>0$.

\medskip

(a) This is standard consequence of the Floquet-Bloch theory, see e.g. \cite[prop.~3.1~(i)]{Raym95} as well as \cite[prop.~4.1]{BaBeBiRaTh24}.

\medskip

(b) This is proven in \cite[lem.~7.3~(b)]{BBL24}. 
In fact, it asserts that for each $k\in\N$, there are exactly $q$ spectral bands $I_1^k,\ldots,I_q^k$ in $\sigma_{r_k}(V)$ of type $B$ implying $I_j^k\not\subseteq \sigma_r(V)$ and all other spectral bands $I$ in $\sigma_{r_k}(V)$ are of type $A$ implying $I\subseteq\sigma_r(V)$. 
The inclusion $I_j^{(k+1)}\subseteq I_j^k$ for all $1\leq j\leq q$ is a consequence of the tower property for type $B$ bands \cite[Property~(B1)]{BBL24}.

\medskip

(c) Define the infinite continued fraction expansion $[0,a_0,a_1,\ldots,a_n,k,1,1,1,\ldots]$ with unique associated irrational number $\alpha$ and its spectral $\alpha$-tree with partial order $\prec$ on the vertex set defined in \cite[def.~1.4]{BBL24}.

Let $n\in\N$ be odd and $c_r\neq [0,0,1]$. Then a short induction using the construction of the spectral $\alpha$-tree implies that there are $q$ vertices $u_1,\ldots,u_q$ in level $n$ of the tree and $q$ vertices $v_1, \ldots, v_q$ in level $n+1$ such that 
\begin{itemize}
\item $u_1$ has label $A$ (see \cite{BBL24} for details), 
\item $v_1,\ldots,v_q$ have label $B$ and
\item $v_1\prec u_1\prec v_2 \prec \ldots \prec v_q\prec u_q$.
\end{itemize}
Using the map $\Psi$ defined in \cite[thm.~1.7]{BBL24}, we get $\Psi(u_j)=J_j$ and $\Psi(v_j)=I_j^k$ for all $1\leq j\leq q$ such that $I_j^k\prec J_j \prec I_{j+1}^k$. 

The case $n\in\N$ even and $c_r\neq [0,0]$ is treated similarly but $u_1$ has label $B$ and $u_1\prec v_1 \prec u_2 \prec \ldots \prec u_q\prec v_q$. For the cases $c_r=[0,0]$ and $c_r=[0,0,1]$, (c) is proven in \cite[lem.~5.3]{BBL24}. 

\medskip

(d) 
The spectrum of these operators can be described by certain polynomials, which are traces of so-called transfer matrices. 
The reader is referred to \cite[sec.~3, app.~II]{BaBeBiRaTh24} for a detailed elaboration. 
We only summarize the necessary properties needed here.
Let $c$ be a finite continued fraction expansion and $s\in[0,1]\cap\QM$ be its evaluation. 
Then there is polynomial $t_c:\RM\to\RM$ (depending on $V$, which we suppress) such that $\sigma_s(V)=\{E\in\RM\,|\, t_c(E)\in[-2,2]\}$, see e.g. \cite[prop.~4.10]{BBL24}. 
Furthermore, $E\in\sigma_c(V)$ is a spectral band edge if and only if $t_c(E)\in\{-2,+2\}$. 
First note that if $V>4$, then $I_j^k\cap \sigma_r(V)=\emptyset$ for all $k\in\N$ by \cite[prop.~4.7, thm.~4.22]{BaBeBiRaTh24}. 
Therefore there is no loss of generality in assuming $V\leq 4$.

Let $E:=E(V)\in\sigma_{r}(V)$ be a spectral band edge of $\sigma_{r}(V)$, i.e. $|t_{c_r}(E)|= 2$. 
We will show that there is a $k_E\in\N$ such that 
$$
\big|t_{d_k}(E)\big|>2, 
	\qquad k\geq k_E.
$$
Suppose for a moment we have proven this. 
Define 
$$
k_0:= \max\{ k_E \,|\, E \textrm{ spectral band edge of } \sigma_{r}(V)\},
$$ 
which is finite since we only have finitely many spectral band edges in $\sigma_{r}(V)$ by (a). 
Thus, for $k\geq k_0$, $|t_{d_{k}}(E)|>2$ follows for all spectral band edges $E$ in $\sigma_{r}(V)$ and so $E\not\in \sigma_{r_{k}}(V)$. 
Hence, (b) yields $I_j^k\cap \sigma_{r}(V)=\emptyset$ for all $k\geq k_0$ finishing the proof.

We proceed proving the claimed estimate for a fixed spectral band edge $E:=E(V)\in\sigma_{r}(V)$.
Recall that $c_r:=[0,a_0, a_1, \dots , a_{n}]$ is the long or short continued fraction expansion of $r$. Define the string 
$$
d_0:= 
	\begin{cases}
		[0],\qquad &\textrm{ if } c_r=[0,0],\\
		[0,a_0, a_1, \dots , a_{n-1}], \qquad &\textrm{ else},
	\end{cases}
$$
with evaluation $r_0$ and spectrum $\sigma_{r_0}(V)$. 
Note that if $d_0=[0]$, we set $r_0=\infty$ and $\sigma_{r_0}(V)=\RM$ and $t_{d_0}(E)=2$, see a discussion in \cite[sec.~2.2]{BBL24} and in \cite[exam.~3.4]{BaBeBiRaTh24}.


Since $|t_{c_r}(E)|= 2$, the statements \cite[lem.~II.1~(b), prop.~II.2~(c)]{BBL24} (see also \cite[lem.~III.2~(b), lem.~3.9]{BaBeBiRaTh24}) imply
\begin{equation}
\label{eq:TraceEstimate_+2}
t_{c_r}(E)=\pm2 
\qquad\Rightarrow\qquad
\big|t_{d_{k+1}}(E)\big| 
	= k \left| \left(1+\frac{1}{k}\right) t_{d_1}(E) \mp t_{d_0}(E) \right|.
\end{equation}

\underline{Case 1:} We treat the case $t_{c_r}(E)=2$. We wish to show $t_{d_1}(E(V))\neq t_{d_0}(E(V))$ for all $V\in(0,4]$.
Since the maps
$$
V\mapsto E(V), \qquad
V\mapsto t_{d_1}(E(V)), \qquad\textrm{and}\qquad
V\mapsto t_{d_0}(E(V))
$$
are continuous \cite[cor.~3.2]{BBL24}, assume towards contradiction that there is a $V_0>0$ such that 
$$
C_0:=t_{d_1}(E(V_0))=t_{d_0}(E(V_0)) 
\qquad\textrm{and}\qquad
t_{d_1}(E(V)) \neq t_{d_0}(E(V)) \textrm{ for all } V>V_0.
$$ 
Then the Fricke-Voigt invariant \cite[prop.~2.3]{Raym95} (see also \cite[prop.~II.2~(b)]{BBL24} and \cite[prop.~3.13]{BaBeBiRaTh24}) and $t_{c_r}(E)=2$ lead to
$$
4+ V^2
	= 4 + t_{d_1}(E(V))^2 + t_{d_0}(E(V))^2 - 2t_{d_1}(E(V))t_{d_0}(E(V)).
$$
Note that $C_0=0$ would imply $V_0=0$ and so $C_0\neq 0$ can be assumed. 
Since $t_{d_1}(E(V_0))=t_{d_0}(E(V_0))$ and the maps are continuous in $V$, there is an $\varepsilon>0$ such that for all $V_0<V<V_0+\varepsilon$, the sign of $t_{d_1}(E(V))$ and $t_{d_0}(E(V))$ agree. 
Thus, the Fricke-Voigt invariant yields for $V_0<V<V_0+\varepsilon$ that
$$
V^2 = |t_{d_1}(E(V))|^2 + |t_{d_0}(E(V))|^2 - 2|t_{d_1}(E(V))| \ |t_{d_0}(E(V))|
	= \big( |t_{d_1}(E(V))| - |t_{d_0}(E(V))|\big)^2.
$$
Sending $V$ from above to $V_0$ implies that $V_0=0$, a contradiction. Hence, $t_{d_1}(E(V))\neq t_{d_0}(E(V))$ follows for all $V\in(0,4]$.
Thus, Equation~\eqref{eq:TraceEstimate_+2} implies that there is $k_E\in\N$ such that 
$$
\big|t_{d_{k+1}}(E)\big| \geq k \left| \left(1+\frac{1}{k}\right) t_{d_1}(E) - t_{d_0}(E) \right| > 2, \qquad k\geq k_E.
$$

\underline{Case 2:} The case $t_{c_r}(E)=-2$ is treated similarly by first proving that $t_{d_1}(E(V))\neq-t_{d_0}(E(V))$ for all $V\in(0,4]$. 
Then, Equation~\eqref{eq:TraceEstimate_+2} implies that there is $k_E\in\N$ such that 
$$
\big|t_{d_{k+1}}(E)\big| \geq k \left| \left(1+\frac{1}{k}\right) t_{d_1}(E) + t_{d_0}(E) \right| > 2, \qquad k\geq k_E.
$$
This finishes the proof.
\end{proof}

\begin{rem}
For a finite string $c$ and its evaluation $s\in[0,1]$, define $\sigma_{c}(V):=\sigma_{s}(V)$ for $V\in\RM$. If $c=[0]$, then set $\sigma_{[0]}(V)=\RM$  as described in \cite{BBL24,BaBeBiRaTh24}.
Let $c_1:=[a_{-1},a_0,\ldots,a_{n-1}]$, $c_2:=[a_{-1},a_0,\ldots,a_n]$ and $c_k:=[a_{-1},a_0,\ldots,a_n,k]$ for $k,n\in\N$.
We actually proved in Lemma~\ref{lem:Spectrum_event_disjoint}~(d) that for all $V\neq 0$, there exists a $k_0\in\N$ such that 
$$
\sigma_{c_1}(V)\cap \sigma_{c_2}(V) \cap \sigma_{c_k}(V) = \emptyset \quad \textrm{ for all } k\geq k_0.
$$
This generalizes \cite[prop.~3.1~(iii)]{Raym95}, confer also \cite[prop.~4.7]{BaBeBiRaTh24}. The result should also be compared with \cite[lem.~7.16]{BBL24} from which our proof is motivated. 
\end{rem}

With this at hand, we can specify Proposition~\ref{prop:EssentialSpectrum} for the operators defined in Equation~\eqref{eq:SchroedOper}. 


\begin{proposi}
\label{prop:DiscreteSpectrum=q_points}
Let $V\neq 0$, $r=\frac{p}{q}\in[0,1]\cap\QM$ with $p$ and $q$ coprime and $r_-,r_+\in\ol{[0,1]}_F$ be the left and right limit point as defined in Corollary~\ref{cor:Completion_FareySpace}. 
Consider the $q$ disjoint spectral bands $J_1\prec J_2\prec \ldots \prec J_q$ with $\sigma_{r}(V)=\bigsqcup_{j=1}^q J_j$.
\begin{enumerate}[(a)]
\item If $r=0=\frac{0}{1}$, then there is an $E^+=E^+(V)\in\RM$ such that 
$$
J_1\prec \{E^+\} \qquad \textrm{and} \qquad \sigma_{0_+}(V)=\sigma_{0}(V)\cup\{E^+(V)\}.
$$
\item If $r=1=\frac{1}{1}$, then there is an $E^-=E^-(V)\in\RM$ such that 
$$
\{E^-\} \prec J_1
 \qquad \textrm{and}\qquad 
 \sigma_{1_-}(V)=\sigma_{1}(V)\cup\{E^-(V)\}.
$$
\item If $r\in(0,1)$ with continued fraction expansions $\ca(r)=[0,0,a_1,\ldots,a_n+1]$ and $\cb(r)=[0,0,a_1,\ldots,a_n,1]$, then there are $E_1^{\pm}=E_1^{\pm}(V), \ldots, E_q^{\pm}=E_q^{\pm}(V)\in\RM$ such that\\ 
for $n\in\N$ even,
\begin{align*}
\{E_1^-\} \prec J_1\prec \{E_2^-\} \prec \ldots \prec \{E_q^-\}\prec J_q 
	\quad &\textrm{ and }  \quad
	\sigma_{r_-}(V)=\sigma_{r}(V)\cup\{E_1^-(V),\ldots,E_q^-(V)\},\\[0.1cm]
 J_1\prec \{E_1^+\} \prec J_2 \prec \ldots \prec J_q\prec \{E_q^+\} 
 	\quad  &\textrm{ and }  \quad
 	\sigma_{r_+}(V)=\sigma_{r}(V)\cup\{E_1^+(V),\ldots,E_q^+(V)\},
\end{align*}
and for $n\in\N$ odd,
\begin{align*}
 J_1\prec \{E_1^-\} \prec J_2 \prec \ldots \prec J_q\prec \{E_q^-\}  
 	 \quad &\textrm{ and }  \quad 
 	 \sigma_{r_-}(V)=\sigma_{r}(V)\cup\{E_1^-(V),\ldots,E_q^-(V)\},\\[0.1cm]
\{E_1^+\} \prec J_1\prec \{E_2^+\} \prec \ldots \prec \{E_q^+\}\prec J_q  
	\quad &\textrm{ and }  \quad 
	\sigma_{r_+}(V)=\sigma_{r}(V)\cup\{E_1^+(V),\ldots,E_q^+(V)\}.
\end{align*}
\item We have $\mu(\sigma_r(V))=\mu(\sigma_{r_-}(V))=\mu(\sigma_{r_+}(V))$ where $\mu$ is the Lebesgue measure.
\item The maps of these eigenvalues $(0,\infty)\ni V\mapsto E_j^\pm(V)$ are Lipschitz continuous.
\end{enumerate}
\end{proposi}

According to the previous proposition, $\sigma_{r_{\pm}}(V)$ has exactly one eigenvalue in each bounded spectral gap of $\sigma_{r}(V)$ and one eigenvalue below the infimum or above the supremum of the spectrum $\sigma_{r}(V)$.
The reader is invited to compare this result with Figure~\ref{fig-Kohmoto} where $r_+$ and $r_-$ for $r=\frac{2}{3}$ have exactly three points in the spectral gaps. 
Similarly, $s_+$ and $s_-$ for $s=\frac{1}{4}$ admit exactly four points in the spectral gaps.
Specifically, Proposition~\ref{prop:DiscreteSpectrum=q_points} provides a proof of the numerical observations addressed in \cite[rem.~3]{BIT91}.

\begin{proof}
Statements (a) and (b) follow similarly as (c) and so we only treat (c).

(c)
For $k\in\N$, let $r_k^l\in[0,1]$ be the evaluation of $c_k:=[0, a_1, \dots , a_{n},1,k]$ and $r_k^s\in[0,1]$ be the evaluation of $d_k:=[0, a_1, \dots , a_{n}+1,k]$. 
We only treat the case that $n\in\N$ is even. 
The odd case follows similarly.
Since $n\in\N$ is even, we have $r_k^l < r < r_k^s$, $\lim_F r_k^l = r_-$ and $\lim_F r_k^s = r_+$. 
We continue proving the statement for $\sigma_{r_+}(V)$.
The analog statement for $\sigma_{r_-}(V)$ is proven following the same lines.

For $1\leq j\leq q$, let $I_j^k$ be the spectral band of $\sigma_{r_k^s}(V)$ defined in Lemma~\ref{lem:Spectrum_event_disjoint}~(b). 
Moreover, let $k_0\in\N$ be such that for all $k\geq k_0$, and each $1\leq j\leq q$, we have $I_j^k\cap \sigma_r(V)=\emptyset$, which exists by Lemma~\ref{lem:Spectrum_event_disjoint}~(d).
By Theorem~\ref{theo-SpeCon}, $\lim_{k\to\infty} \sigma_{r_k^s}(V)=\sigma_{r_+}(V)$. 
Thus, $I_j^{k+1}\subseteq I_j^k$ implies that the Hausdorff limit $\lim_{k\to\infty} I_j^k=:A_j$ exists and $A_j\cap \sigma_r(V)=\emptyset$. 
Note that $A_j\cap A_i=\emptyset$ for $i\neq j$ since $I_j^k\cap I_i^k=\emptyset$ for all $k\in\N$ (spectral bands do not touch). 
By construction, we have 
$$
\sigma_{r_+}(V)\setminus \sigma_r(V)
	= \bigsqcup_{j=1}^q A_j.
$$ 
Thus, it is left to prove that $A_j=\{E_j^+\}$ is a single point for each $1\leq j\leq q$. 
Note that $E_j^+=E_j^+(V)$ depends on $V$ since $I_j^k$ does so.
Since $A_j$ is the Hausdorff limit of an interval $I_j^k$ it is either a single point or an interval. 
However, $A_j$ cannot be an interval by Proposition~\ref{prop:EssentialSpectrum} proving $A_j=\{E_j^+\}$.
Hence, $\sigma_{r_+}(V)=\sigma_{r}(V)\cup\{E_1^+(V),\ldots,E_q^+(V)\}$ is concluded.
The ordering $J_1\prec \{E_1^+\} \prec J_2 \prec \ldots \prec J_q\prec \{E_q^+\}$ follows from Lemma~5.5~(c) asserting the corresponding ordering for the spectral bands $I_j^k$ for all $k\in\N$.

(d) This is an immediate consequence of (a),(b) and (c) using $\mu(\{E_1^{\pm}, \ldots, E_q^{\pm}\})=0$.

(e) This is straightforward using that $(0,\infty)\ni V\mapsto\sigma(H_{\omega,V})$ is Lipschitz continuous in the Hausdorff metric, see e.g. \cite[lem.~3.1]{BBL24}.
\end{proof}


\section{Optimality in the Kohmoto model}
\label{sect-SpectralMapContOpt}

We show that the spectral estimates obtained in Theorem~\ref{theo-SpeCon} are optimal if $V > 4$, see Theorem~\ref{theo-Optimality}. 
This optimality result  is similar as for the Almost-Mathieu operator, where the spectral map is $\frac{1}{2}-$Hölder-continuous \cite{AMS}.
Also there optimality can be observed at certain rational points where spectral gaps are closing, see \cite{BeRa90,HaKaTa16} as well as a discussion in \cite{BT21}. 

In order to prove Theorem~\ref{theo-Optimality}, we use the links between the spectra at rational points and their continued fractions, which were first observed in \cite{Raym95} and recently further developed \cite{BBL24}. 
For convenience, we use the notation $\sigma_{x}(V) := \sigma(H_{\omega_x,V})$ for $x\in\ol{[0,1]}_F$ and $V\in\RM$ where $H_{\omega_x,V}$ is the self-adjoint operator defined in Equation~\eqref{eq:SchroedOper} and $\omega_x$ is defined in Definition~\ref{def:omega_x}.

Recall the notion of spectral bands introduced in the previous section. The following lemma collects some more structural properties in the spirit of Lemma~\ref{lem:Spectrum_event_disjoint}.

\begin{lemma}
\label{lem:BasicSpectralBands}
Let $V\neq 0$, $r\in[0,1]\cap\QM$ and $[0, a_0, \dots , a_{n}]$ be its (long or short) continued fraction expansion. 
For $k\in\N$, let $r_k\in[0,1]$ be the evaluation of $[0, a_0, \dots , a_{n},k]$. Then the following assertions hold.
\begin{enumerate}[(a)]
\item Each spectral band $I$ of $\sigma_r(V)$ contains at most $k$ spectral bands of $\sigma_{r_k}(V)$.
\item For all $k\in\N$, we have $\sigma_{r_{k+1}}(V)\subseteq \sigma_{r}(V)\cup\sigma_{r_{k}}(V)$. 
\item If $V>4$, then $\sigma_{r}(V)\cap\sigma_{r_k}(V)\cap\sigma_{r_{k+1}}(V)=\emptyset$ for $k\in\N$. 
\end{enumerate}
\end{lemma}

\begin{proof}
To prove (a), let $I$ be a spectral band of $\sigma_{r}(V)$. 
Then \cite[thm.~1.7]{BBL24} asserts that $I$ contains either $k$ or $k-1$ spectral bands (of type $A$) in $\sigma_{r_k}(V)$ involving the construction of the so-called spectral $\alpha$-tree by choosing $\alpha$ with infinite continued fraction expansion $[0, a_0, \dots , a_{n},k,1,1,1,\ldots]$. 
For $V>4$, this is also proven in \cite[prop.~3.1]{Raym95} and \cite[lem.~4.14]{BaBeBiRaTh24}.

The statements (b) and (c) are proven in \cite[prop.~3.1~(ii),(iii)]{Raym95}, see also \cite[lem.~4.3, prop.~4.7]{BaBeBiRaTh24}.
\end{proof}

\begin{theo}
\label{theo-Optimality}
Let $V > 4$ and $x\in\ol{[0,1]}_F\setminus [0,1]$ and $(H_{\omega,V})_{\omega\in\Aa^\ZM}$ be as defined in Equation~\eqref{eq:SchroedOper}.
Then there are $C_1,C_2>0$ and a sequences $(\alpha_j)_{j\in \N}\subseteq [0,1] \cap \QM$  with $\lim_F \alpha_j = x$ such that 
$$
C_1 d_F(\alpha_j, x)
	\leq d_H(\sigma(H_{\omega_{\alpha_j},V}), \sigma(H_{\omega_x,V}))
	\leq C_2 d_F(\alpha_j, x), \qquad j\in\N.
$$
\end{theo}

\begin{proof}
Let $V>4$ be fixed and $\sigma_x(V):=\sigma(H_{\omega_x,V})$ for $x\in\ol{[0,1]}_F$.
To simplify reading, we omit the $V$ dependence in the notation.
By Corollary~\ref{cor:Completion_FareySpace}, there exists an $r =\frac{p}{q} \in [0,1]\cap \QM$ with $p$, $q$ coprime such that $x\in\{r_-,r_+\}$.
Let $k\in\N$.
If $r=0$, define $r_k$ as the evalutaion of $[0,0,k]$ satisfying $\lim_F r_k=0_+$.
If $r=1$, define $s_k$ as the evalutaion of $[0,0,1,k]$ satisfying $\lim_F s_k=1_-$.
If $r\in(0,1)$, let $\cb=[0,0, a_1, \dots a_n, 1]$ be the long continued fraction expansion of $r$.
Define $r_k$ as the evaluation $[0,0, a_1, \ldots, a_n, 1, k]$ and $s_k$ as the evaluation of $[0,0, a_1, \ldots, a_n + 1, k]$. 
If $n$ is even, then $r_k < r < s_k$, $\lim_F r_k=r_-$ and $\lim_F s_k=r_+$ hold, see  Propositon~\ref{prop-convergence_in_farey_topo}. 
Similarly, if $n$ is odd, then $s_k < r < r_k$, $\lim_F r_k=r_+$ and $\lim_F s_k=r_-$.
We will show that one can extract a suitable subsequence of $(r_k)_{k\in\N}$ respectively $(s_k)_{k\in\N}$ satisfying the claim of the theorem. 
Since the cases are treated similarly, we only focus on the case that $n$ is even and we only extract a subsequence of $(r_k)_{k\in\N}$ satisfying $\lim_F r_k=r_-$.
Thus, Theorem~\ref{theo-SpeCon} implies $\lim_{k\to\infty} d_H\big(\sigma_{r_k},\sigma_{r_-}\big)=0$.


The proof is divided up in various steps using the notation $D_k:=d_H\big(\sigma_{r_k}\cap\sigma_r, \sigma_r\big)$ for $k\in\N$. 
Moreover, for $a\in\RM$ and $B\subseteq \RM$ compact, we define the distance $\dist(a,B):=\inf_{b\in B} |b-a|$.
With this notation fixed, we express the Hausdorff distance by
$$
d_H(A,B)=\max\left\{ \sup_{y\in A}\dist(y,B), \sup_{z\in B}\dist(z,A)\right\}, \qquad A,B\in \Kk(\RM).
$$

\underline{Step 1:} We show that there is a $q'\in\N_0$ such that $\frac{1}{k(q+q')}\leq d_F(r_k ,r_-) \leq \frac{1}{kq}$ for all $k\geq 2$.

Let $k\geq 2$.
Standard properties of continued fraction expansion (see e.g. \cite[lem.~3A]{Sc80}) assert that $r_k = \frac{p_k}{q_k}$ with $p_k$ and $q_k$ coprime and $q_k=kq +q'$ for some $q'\in\N$. 
Thus, $r_k\in F_{kq+q'} \setminus F_{kq+q'-1}$ holds.
Since $n$ is even, Lemma~\ref{lem-FarNghCF} (applied to $[0,0, a_1, \ldots, a_n, 1, k-1,1]$) asserts that $(r_k)_* = r_{k-1}$ and $(r_k)^* = r$ are the Farey neighbors of $r_k$ in $F_{kq+q'}$.

Thus, $r_k,r_- \in \overline{(r_{k-1},r)}_F$ where $r_{k-1},r\in F_m$ are $m$-Farey neighbors for $m=kq +q'-1$ and $r_k\in F_{m+1}$. 
Hence, 
$$ 
d_F(r_k ,r_- ) 
	= \frac{1}{m+1}
	= \frac{1}{kq + q'} 
	\begin{cases}
		&\leq \frac{1}{kq},\\
		&\geq \frac{1}{k(q+q')},
	\end{cases}
	\qquad k\in\N.
$$

Note that if $r_k$ is defined as the evaluation of $[0,0,k]$ and $r=\frac{0}{1}=\frac{p}{q}$, then $d_F(r_k ,0_+ )= \frac{1}{kq}$ holds since in this case $q_k= (k-1)q+1=k$.
\medskip

\underline{Step 2:} We show that there exists an $l_0\in\N$ such that $D_k\leq d_H\big(\sigma_{r_k},\sigma_{r_-}\big)$ for all $k\geq l_0$.

By Proposition~\ref{prop:DiscreteSpectrum=q_points}, we have $\sigma_{r_-}\setminus\sigma_r = \{E_1,\ldots,E_q\}$. 
Let $k_0\in\N$ be chosen as in Lemma~\ref{lem:Spectrum_event_disjoint}~(d). 
Then $\sigma_{r_{k_0}}\setminus\sigma_r$ is compact by Lemma~\ref{lem:Spectrum_event_disjoint}~(c) and (d).
Thus, 
$$
\varepsilon 
	:= \frac{1}{2} \min\left\{ 
			\min_{1\leq j\leq q} \dist\big(E_j,\sigma_r\big),
			\inf_{E\in\sigma_{r_{k_0}}\setminus\sigma_r} \dist\big(E,\sigma_r\big)
		 \right\}
		 >0.
$$
Choose $l_0\geq \max\{k_0,2\}$ such that $d_H\big(\sigma_{r_k},\sigma_{r_-}\big)<\varepsilon$ for $k\geq l_0$. 
Let $0<\delta<\frac{\varepsilon}{2}$ and $k\geq l_0$. 

If $y\in\sigma_r\subseteq \sigma_{r_-}$ (inclusion holds by Proposition~\ref{prop:EssentialSpectrum}), then there exists a $z\in \sigma_{r_k}$ such that $|y-z|<d_H\big(\sigma_{r_k},\sigma_{r_-}\big)+\delta$ by the definition of the Hausdorff metric. 
Since $k\geq l_0$, we conclude 
$$
|y-z|
	< \frac{3}{4}\inf_{E\in\sigma_{r_{k_0}}\setminus\sigma_r} \dist\big(E,\sigma_r\big).
$$ 
By Lemma~\ref{lem:BasicSpectralBands}~(b), we have $z\in\sigma_{r_k}\subseteq \sigma_r\cup\sigma_{r_{k_0}}$.
Hence, the previous estimate and $y\in\sigma_r$ lead to $z\in \sigma_r$. 
Thus, $z\in\sigma_r\cap\sigma_{r_k}$ and so
$$
\dist\big(y,\sigma_{r_k}\cap\sigma_r\big)
	= \inf_{\tilde{z}\in \sigma_{r_k}\cap\sigma_r} |y-\tilde{z}|
	< d_H\big(\sigma_{r_k},\sigma_{r_-}\big)+\delta,
	\qquad y\in\sigma_r.
$$
If $z\in \sigma_{r_k}\cap\sigma_r$, then there exists a $y\in\sigma_{r_-}$ such that $|y-z|<d_H\big(\sigma_{r_k},\sigma_{r_-}\big)+\delta$. 
Since $k\geq k_0$, we conclude $|y-z|< \frac{3}{4}\min_{1\leq j\leq q} \dist\big(E_j,\sigma_r\big)$ implying $y\neq E_j$ for all $1\leq j\leq q$. 
Thus, $y\in \sigma_r$ follows leading to
$$
\dist\big(z,\sigma_r\big)
	= \inf_{\tilde{y}\in \sigma_r} |\tilde{y}-z|
	< d_H\big(\sigma_{r_k},\sigma_{r_-}\big)+\delta
	,\qquad z\in \sigma_{r_k}\cap\sigma_r.
$$
Since $\delta>0$ was arbitrary, the previous considerations lead to
$$
D_k=d_H\big(\sigma_{r_k}\cap\sigma_r, \sigma_r\big)
	\leq d_H\big(\sigma_{r_k},\sigma_{r_-}\big)
	, \qquad k\geq k_0.
$$

\medskip

\underline{Step 3:} We prove 
$$
\mu\big(\sigma_r\cap\sigma_{r_k}\big) 
	\leq \mu\big(\sigma_r\big) 
	\leq \mu\big(\sigma_r\cap\sigma_{r_k}\big) + 2q(k+1) D_k
	, \qquad k\in\N,
$$
where $\mu(B)$ denotes the Lebesgue measure of $B\subseteq \RM$. 

The first inequality follows trivially by monotonicity of the Lebesgue measure. 
Let $k\in\N$ and define $G_k := \{I : I \text{ is a connected component of } \sigma_r\setminus \sigma_{r_k}\}$. 
By definition of the Hausdorff metric, we conclude $2D_k \geq \max_{I \in G_k} \mu(I)$. 
Then Lemma~\ref{lem:Spectrum_event_disjoint}~(a) and Lemma~\ref{lem:BasicSpectralBands}~(a) imply $\sharp G_k \leq q(k+1)$. 
Hence,
$$ 
\mu\big(\sigma_r \setminus \sigma_{r_k}\big)
	= \sum_{I \in G_k} \mu(I) 
	\leq q(k+1) \max_{I \in G_k} \mu(I)
	\leq 2q(k+1) D_k,
$$
follows and implies
$$ 
\mu\big(\sigma_r\big)
	= \mu\big(\sigma_r \cap \sigma_{r_k}\big) + \mu\big(\sigma_r \setminus \sigma_{r_k}\big)
	\leq \mu\big(\sigma_r \cap \sigma_{r_k}\big) + 2q(k+1)D_k.
$$

\medskip

\underline{Step 4:} Using $V>4$, we prove for all $k\in\N$,
$$
\max\left\{ (k+1)D_k, (k+2)D_{k+1} \right\}
	\geq \frac{\mu(\sigma_r(V))}{4q}>0.
$$

First note that $\frac{\mu(\sigma_r)}{4q}$ is nonzero since the spectrum $\sigma_r$ is the disjoint union of $q$ intervals of positive length, see Lemma~\ref{lem:Spectrum_event_disjoint}~(a).
Since $V>4$, Lemma~\ref{lem:BasicSpectralBands}~(c) implies that the sets $\sigma_r\cap\sigma_{r_k}$ and $\sigma_r\cap\sigma_{r_{k+1}}$ are disjoint subsets of $\sigma_r$. 
Thus, Step 3 leads to
\begin{align*}
4q \max\left\{ (k+1)D_k, (k+2)D_{k+1} \right\} 
	\geq &2q(k+1)D_k + 2q(k+2)D_{k+1}\\
	\geq &2\mu\big(\sigma_r\big) - \mu\big(\sigma_r\cap\sigma_{r_k}\big) - \mu\big(\sigma_r\cap\sigma_{r_{k+1}}\big)\\
	\geq &\mu\big(\sigma_r\big).
\end{align*}

\medskip

\underline{Step 5:} Step 4 implies that there is a subsequence $(k_j)_{j\in\N}$ (i.e. $k_j<k_{j+1}$) such that 
$$
(k_j+1)D_{k_j}
	\geq \frac{\mu(\sigma_r)}{4q}
	>0
	,\qquad j\in\N.
$$ 
Thus, there exists a $j_0\in\N$ such that $k_j\geq \max\{l_0,2\}$ (where $l_0$ is the integer chosen in Step 2) and $k_jD_{k_j}\geq \frac{\mu(\sigma_r)}{8q}>0$ for all $j\geq j_0$. 
Set $C_1:=\frac{\mu(\sigma_r)}{8}$. 
Then Step 1 leads to 
$$
D_{k_j}
	\geq C_1 \frac{1}{q k_j}
	\geq C_1 d_F(r_{k_j},r_-)
	, \qquad j\geq j_0.
$$ 
On the other hand, Theorem~\ref{theo-SpeCon_PatternEquiv} asserts that there is a constant $C_2>0$ such that 
$$
d_H\big(\sigma_{r_{k_j}},\sigma_{r_-}\big) \leq C_2 d_F(r_{k_j},r_-)
$$ 
for $j\in\N$. 
Invoking Step 2 (using $k_j\geq l_0$ for $j\geq j_0$), we  conclude
\begin{equation}
\label{eq:Optimality_SpectralDistance_Estimates}
C_1 d_F(r_{k_j},r_-) 
	\leq D_{k_j} 
	\leq d_H\big(\sigma_{r_{k_j}},\sigma_{r_-}\big) 
	\leq C_2 d_F(r_{k_j},r_-)
	, \qquad j\geq j_0.
\end{equation}
This proves the claimed statement for $\alpha_j:=r_{k_{j+j_0}}$ and $j\in\N$ since $\lim_F\alpha_j=\lim_F r_k=r_-$.
%
%
\end{proof}

\begin{rem}
The assumption $V>4$ is only used to have that
$
\sigma_{r_k}(V)\cap \sigma_{r}(V)
$
and
$
\sigma_{r_{k+1}}(V)\cap\sigma_{r}(V)
$
are disjoint, which is used in Step 4.
\end{rem}

From the previous considerations, we can actually conclude the following observations for the convergence of the Lebesgue measure at the rational points.

\begin{coro}
\label{cor:Measure_Rate_RationalAppr}
Let $V\neq 0$, $r\in[0,1]\cap\QM$ and $[0,a_0, a_1, \dots , a_{n}]$ be its (long or short) continued fraction expansion. For $k\in\N$, let $r_k\in[0,1]$  be the evaluation of $[0,0, a_1, \dots , a_{n},k]$. 
\begin{enumerate}[(a)]
\item If there is a subsequence $(k_j)_j$ such that 
$$\lim_{j\to\infty}k_j d_H\big(\sigma_{r_{k_j}}(V)\cap\sigma_{r}(V),\sigma_{r}(V)\big)=0,
$$ 
then the Lebesgue measure of these sets converge, i.e.
	$\lim\limits_{j\to\infty} \mu\big(\sigma_{r}(V)\cap\sigma_{r_{k_j}}(V)\big)=\mu\big(\sigma_{r}(V)\big)$.
\item If $V>4$, then there is a subsequence $(k_j)_{j\in\N}$ such that 
	$$
	\lim\limits_{j\to\infty} \mu\big(\sigma_{r}(V)\cap\sigma_{r_{k_j}}(V)\big)\leq \frac{1}{2}\mu\big(\sigma_{r}(V)\big).
	$$
\item If $V>4$, then there are $C_1,C_2>$ and a subsequence $(k_j)_{j\in\N}$ such that 
$$
\frac{C_1}{k_j} 
	\leq d_H\big(\sigma_{r_{k_j}}(V)\cap\sigma_{r}(V),\sigma_{r}(V)\big)
	\leq \frac{C_2}{k_j}, \qquad j\in\N.
$$
\item If $V>4$ and there is a subsequence $(k_j)_{j\in\N}$ such that $\lim\limits_{j\to\infty}\mu\big(\sigma_{r}(V)\cap\sigma_{r_{k_j}}(V)\big)=\mu\big(\sigma_{r}(V)\big)$, then 
	$\lim\limits_{j\to\infty}\mu\big(\sigma_{r}(V)\cap\sigma_{r_{k_j+1}}(V)\big)=0$. 
\end{enumerate}
\end{coro}

\begin{proof}
For simplification, the $V$ dependence is dropped in the notation and as before we set $D_k:=d_H\big(\sigma_{r_k}\cap\sigma_{r}, \sigma_{r}\big)$ for $k\in\N$.

(a) According to Step 3 in the proof of Theorem~\ref{theo-Optimality}, we have for $k\in\N$,
\begin{align*}
\mu\big(\sigma_{r}\cap\sigma_{r_k}\big) \leq \mu\big(\sigma_{r}\big) 
	\leq \mu\big(\sigma_{r}\cap\sigma_{r_k}\big) + 2q(k+1) D_k
	\leq \mu\big(\sigma_{r}\cap\sigma_{r_k}\big) + 4q \ k D_k.
\end{align*}
This immediately yields the claim in (a).

\medskip

(b) Since $V>4$, Lemma~\ref{lem:BasicSpectralBands}~(c) implies that the sets $\sigma_{r}\cap\sigma_{r_k}$ and $\sigma_{r}\cap\sigma_{r_{k+1}}$ are disjoint subsets of $\sigma_{r}$. Thus, for $k\in\N$,
\begin{align*}
2\min \left\{ \mu\big(\sigma_{r}\cap\sigma_{r_k}\big), \, \mu\big(\sigma_{r}\cap\sigma_{r_{k+1}}\big) \right\}
	\leq \mu\big(\sigma_{r}\cap\sigma_{r_k}\big) + \mu\big(\sigma_{r}\cap\sigma_{r_{k+1}}\big) 
	\leq \mu\big(\sigma_{r}\big).
\end{align*}
Hence, there is at least one subsequence $(k_j)_j$ such that $\mu\big(\sigma_{r}\cap\sigma_{r_{k_j}}\big)\leq \frac{1}{2}\mu\big(\sigma_{r}\big)$ and by passing to another subsequence there is no loss of generality in assuming that the limit 
$$
\lim_{j\to\infty} \mu\big(\sigma_{r}\cap\sigma_{r_{k_j}}\big)\leq \frac{1}{2}\mu\big(\sigma_{r}\big)
$$ 
exists.

\medskip

(c) This follows from Equation~\eqref{eq:Optimality_SpectralDistance_Estimates} and step 1 in the proof of Theorem~\ref{theo-Optimality}.

\medskip

(d) As in (b), $V>4$ implies 
$$
\mu\big(\sigma_{r}\cap\sigma_{r_k}\big) + \mu\big(\sigma_{r}\cap\sigma_{r_{k+1}}\big) 
	\leq \mu\big(\sigma_{r}\big).
$$
Hence, if $(k_j)_j$ satisfies $\lim\limits_{j\to\infty}\mu\big(\sigma_{r}\cap\sigma_{r_{k_j}}\big)=\mu\big(\sigma_{r}\big)$, then $\lim\limits_{j\to\infty}\mu\big(\sigma_{r}\cap\sigma_{r_{k_j+1}}\big)=0$ follows.
\end{proof}

\begin{rem}
Due to \cite{BIST89}, we have $\mu\big( \sigma(H_{\alpha,V})\big)=0$ for all irrational $\alpha\in[0,1]\setminus\QM$. 
On the other hand, $\mu\big( \sigma_x(V) \big)>0$ holds for all $x\in\ol{[0,1]}_F\setminus\big([0,1]\setminus\QM\big)$ using Proposition~\ref{prop:DiscreteSpectrum=q_points}~(d). 
Note that it is straightforward to prove that the Hausdorff convergence $\lim\limits_{n\to\infty} \sigma_{x_n}(V)=\sigma_{x}(V)$ implies
$$
\limsup_{n\to\infty} \mu\big( \sigma_{x_n}(V)\big) \leq \mu \big(\sigma_{x}(V)\big).
$$
Thus, the map $$
\ol{[0,1]}_F\ni x\mapsto \mu\big(\sigma(H_{\omega_x,V})\big)
$$
is continuous at all $x\in[0,1]$ irrational and else discontinuous by Corollary~\ref{cor:Measure_Rate_RationalAppr}~(b) if $V>4$.
\end{rem}

%
%
%

\printbibliography

\end{document}